\def\year{2019}\relax
\documentclass[letterpaper]{article}
\pdfoutput=1
\usepackage{aaai19}
\usepackage{times}
\usepackage{helvet}
\usepackage{courier}
\usepackage{url}
\usepackage{graphicx}
\title{Solving Large Extensive-Form Games with Strategy Constraints}
\author{Trevor Davis,\textsuperscript{1}
Kevin Waugh,\textsuperscript{2}
Michael Bowling,\textsuperscript{2,1}\\
\textsuperscript{1}University of Alberta\\
\textsuperscript{2}DeepMind\\
trdavis1@ualberta.ca,
\{waughk, bowlingm\}@google.com}

\usepackage[utf8]{inputenc}
\usepackage[T1]{fontenc}
\usepackage{booktabs}
\usepackage{amsfonts}
\usepackage{nicefrac} 
\usepackage{microtype}

\usepackage{comment}
\usepackage{times}
\usepackage{xcolor}
\usepackage{soul}
\usepackage[small,skip=5pt]{caption}

\usepackage{algorithm,algpseudocode}
\usepackage{amsmath,amssymb,amsthm}
\usepackage{bm}
\usepackage{color}
\usepackage{mathtools}
\usepackage{multicol}
\usepackage{siunitx}
\usepackage[skip=0pt]{subcaption}
\usepackage{varwidth}

\newcommand{\vect}[1]{{\bm{\mathbf{#1}}}}
\newcommand{\mat}[1]{\mathbf{#1}}

\newcommand{\pub}{s}
\newcommand{\our}{c_i}
\newcommand{\opp}{c_{-i}}
\newcommand{\Priv}{C}
\newcommand{\Pub}{S}
\newcommand{\Pubterm}{S_Z}
\newcommand{\Pubtermu}{S_Z^u}

\DeclareMathOperator*{\argmin}{argmin}

\newcommand{\ourstrat}{\sigma_i}
\newcommand{\oppstrat}{\sigma_{-i}}
\newcommand{\oppstratp}{\sigma_{-i}'}
\newcommand{\feas}{\sigma_f}

\newtheorem{lem}{Lemma}

\newtheorem{thm}{Theorem}
\newtheorem{cor}{Corollary}[thm]

\usepackage{tikz}

\usetikzlibrary{positioning}

\newcommand{\tgScheme}{
\scalebox{0.75}{
\begin{tikzpicture}[circ/.style={circle,draw,minimum size=13,inner sep=0pt},node distance=0.5]
  \foreach \x in {0,...,5}
    \foreach \y in {0,...,2}
       \node [circ]  (\x\y) at (\x,\y) {};

  \foreach \x in {0,...,5}
    \foreach \y [count=\yi] in {0,...,1}  
      \draw (\x\y)--(\x\yi);
      
  \foreach \x [count=\xi] in {0,...,4}
    \foreach \y in {0,...,2}  
      \draw (\x\y)--(\xi\y);
  
  \foreach \x [count=\xi] in {0,...,4}
    \foreach \y [count=\yi] in {0,...,1}  
      \draw (\x\y)--(\xi\yi) (\x\yi)--(\xi\y);
  
  \node [circ]  (start) [left=of 01] {$s^0_e$};
  \foreach \y in {0,...,2}
    \draw [->] (start)--(0\y);
  \node [circ,fill=gray!40]  at (20) {$s^0_p$};
  
  \node[coordinate] (l) at (02.west)[yshift=12pt] {};
  \node[coordinate] (r) at (52.east)[yshift=12pt] {};
  \draw [<->,thick]  (l)-- node[above]{$2w$} (r);
  \node[coordinate] (b) at (50.south)[xshift=12pt] {};
  \node[coordinate] (t) at (52.north)[xshift=12pt] {};
  \draw [<->,thick]  (b)-- node[right]{$w$} (t);
\end{tikzpicture}}
}

\begin{document}
\maketitle

\begin{abstract}
Extensive-form games are a common model for multiagent interactions with imperfect information. In two-player zero-sum games, the typical solution concept is a Nash equilibrium over the unconstrained strategy set for each player. In many situations, however, we would like to constrain the set of possible strategies. For example, constraints are a natural way to model limited resources, risk mitigation, safety, consistency with past observations of behavior, or other secondary objectives for an agent. In small games, optimal strategies under linear constraints can be found by solving a linear program; however, state-of-the-art algorithms for solving large games cannot handle general constraints. In this work we introduce a generalized form of Counterfactual Regret Minimization that provably finds optimal strategies under any feasible set of convex constraints. We demonstrate the effectiveness of our algorithm for finding strategies that mitigate risk in security games, and for opponent modeling in poker games when given only partial observations of private information.
\end{abstract}

\section{Introduction}

Multiagent interactions are often modeled using \emph{extensive-form games} (EFGs), a powerful framework that incoporates sequential actions, hidden information, and stochastic events. Recent research has focused on computing approximately optimal strategies in large extensive-form games, resulting in a solution to heads-up limit Texas Hold'em, a game with approximately $10^{17}$ states \cite{Bowling15}, and in two independent super-human computer agents for the much larger heads-up no-limit Texas Hold'em \cite{Moravcik17,Brown18}.

When modeling an interaction with an EFG, for each outcome we must specify the agents' utility, a cardinal measure of the outcome's desirability. Utility is particularly difficult to specify. Take, for example, situations where an agent has multiple objectives to balance: a defender in a security game with the primary objective of protecting a target and a secondary objective of minimizing expected cost, or a robot operating in a dangerous environment with a primary task to complete and a secondary objective of minimizing damage to itself and others. How these objectives combine into a single value, the agent's utility, is ill-specified and error prone.

One approach for handling multiple objectives is to use a linear combination of per-objective utilities. This approach has been used in EFGs to ``tilt'' poker agents toward taking specific actions \cite{Johanson11}, and to mix between cost minimization and risk mitigation in sequential security games \cite{Lisy16}. However, objectives are typically measured on incommensurable scales. This leads to dubious combinations of weights often selected by trial-and-error.

A second approach is to constrain the agents' strategy spaces directly. For example, rather than minimizing the expected cost, we use a hard constraint that disqualifies high-cost strategies. Using such constraints has been extensively studied in single-agent perfect information settings \cite{Altman99} and partial information settings \cite{Isom08,Poupart15,Santana16}, as well as in (non-sequential) security games \cite{Brown14}.

Incorporating strategy constraints when solving EFGs presents a unique challenge. Nash equilibria can be found by solving a linear program (LP) derived using the sequence-form representation \cite{Koller96}. This LP is easily modified to incorporate linear strategy constraints; however, LPs do not scale to large games. Specialized algorithms for efficiently solving large games, such as an instantiation of Nesterov's \emph{excessive gap technique} (EGT) \cite{Hoda10} as well as \emph{counterfactual regret minimization} (CFR) \cite{Zinkevich07} and its variants \cite{Lanctot09,Tammelin15}, cannot integrate arbitrary strategy constraints directly. Currently, the only large-scale approach is restricted to constraints that consider only individual decisions \cite{Farina17}.

In this work we present the first scalable algorithm for solving EFGs with arbitrary convex strategy constraints. Our algorithm, Constrained CFR, provably converges towards a strategy profile that is minimax optimal under the given constraints. It does this while retaining the $\mathcal{O}(1/\sqrt{T})$ convergence rate of CFR and requiring additional memory proportional to the number of constraints. We demonstrate the empirical effectiveness of Constrained CFR by comparing its solution to that of an LP solver in a security game. We also present a novel constraint-based technique for opponent modeling with partial observations in a small poker game.

\section{Background}\label{sec:bg}

Formally, an extensive-form game \cite{Osborne94} is a game tree defined by:
\begin{itemize}
\item A set of \emph{players} $N$. This work focuses on games with two players, so $N=\{1,2\}$.
\item A set of \emph{histories} $H$, the tree's nodes rooted at $\emptyset$. The leafs, $Z \subseteq H$, are \emph{terminal histories}. For any history $h \in H$, we let $h' \sqsubset h$ denote a prefix $h'$ of $h$, and necessarily $h' \in H$.
\item For each $h \in H \setminus Z$, a set of \emph{actions} $A(h)$. For any $a \in A(h)$, $ha \in H$ is a child of $h$.
\item A \emph{player function} $P\colon H\setminus Z \to N \cup \{c\}$ defining the player to act at $h$. If $P(h)=c$ then \emph{chance} acts according to a known probability distribution $\sigma_c(h) \in \Delta_{|A(h)|}$, where $\Delta_{|A(h)|}$ is the probability simplex of dimension $|A(h)|$.
\item A set of \emph{utility functions} $u_i \colon Z \to \mathbb{R}$, for each player.  Outcome $z$ has utility $u_i(z)$ for player $i$. We assume the game is \emph{zero-sum}, i.e., $u_1(z) = -u_2(z)$. Let $u(z) = u_1(z)$.
\item For each player $i \in N$, a collection of \emph{information sets} $\mathcal{I}_i$. $\mathcal{I}_i$ partitions $H_i$, the histories where $i$ acts.  Two histories $h,h'$ in an information set $I\in\mathcal{I}_i$ are indistinguishable to $i$. Necessarily $A(h)=A(h')$, which we denote by $A(I)$. When a player acts they do not observe the history, only the information set it belongs to, which we denote as $I[h]$.
\end{itemize}
We assume a further requirement on the information sets $\mathcal{I}_i$ called \emph{perfect recall}.  It requires that players are never forced to forget information they once observed.  Mathematically this means that all indistinguishable histories share the same sequence of past information sets and actions for the actor. Although this may seem like a restrictive assumption, some perfect recall-like condition is needed to guarantee that an EFG can be solved in polynommial time, and all sequential games played by humans exhibit perfect recall.

\subsection{Strategies}

A \emph{behavioral strategy} for player $i$ maps each information set $I \in \mathcal{I}_i$ to a distribution over actions, $\sigma_i(I) \in \Delta_{|A(I)|}$. The probability assigned to $a \in A(I)$ is $\sigma_i(I,a)$. A \emph{strategy profile}, $\sigma = \{\sigma_1,\sigma_2\}$, specifies a strategy for each player. We label the strategy of the opponent of player $i$ as $\sigma_{-i}$. The sets of behavioral strategies and strategy profiles are $\Sigma_i$ and $\Sigma$ respectively.

A strategy profile uniquely defines a \emph{reach probability} for any history $h \in H$:
\begin{equation}
\pi^{\sigma}(h) \coloneqq \prod_{h'a \sqsubseteq h} \sigma_{P(h')}(I[h'],a)  \label{eq:reach}
\end{equation}
This product decomposes into contributions from each player and chance, $\pi^{\sigma_1}_1(h) \pi^{\sigma_2}_2(h) \pi_c(h)$. For a player $i\in N$, we denote the contributions from the opponent and chance as $\pi^{\sigma_{-i}}_{-i}(h)$ so that $\pi^{\sigma}(h) = \pi_i^{\sigma_i}(h) \pi_{-i}^{\sigma_{-i}}(h)$. By perfect recall we have $\pi^{\sigma_i}_i(h) = \pi^{\sigma_i}_i(h')$ for any $h,h'$ in same information set $I\in\mathcal{I}_i$. We thus also write this probability as $\pi^{\sigma_i}_i(I)$.

Given a strategy profile $\sigma=\{\sigma_1,\sigma_2\}$, the expected utility for player $i$ is given by
\begin{equation}
u_i(\sigma)=u_i(\sigma_1,\sigma_2) \coloneqq \sum_{z \in Z} \pi^{\sigma}(z) u_i(z).
\end{equation}

A strategy $\sigma_i$ is an \emph{$\varepsilon$-best response} to the opponent's strategy $\sigma_{-i}$ if $u_i(\sigma_i,\sigma_{-i}) + \varepsilon \ge u_i(\sigma_i', \sigma_{-i})$ for any alternative strategy $\sigma_i'\in\Sigma_i$.  A strategy profile is an \emph{$\varepsilon$-Nash equilibrium} when each $\sigma_i$ is a $\varepsilon$-best response to its opponent; such a profile exists for any $\varepsilon\ge 0$. The \emph{exploitability} of a strategy profile is the smallest $\varepsilon=\nicefrac{1}{2}(\varepsilon_1+\varepsilon_2)$ such that each $\sigma_i$ is an $\varepsilon_i$-best response. Due to the zero-sum property, the game's Nash equilibria are the saddle-points of the minimax problem
\begin{equation}
\max_{\sigma_1 \in \Sigma_1} \min_{\sigma_2 \in \Sigma_2} u(\sigma_1,\sigma_2) = \min_{\sigma_2 \in \Sigma_2} \max_{\sigma_1 \in \Sigma_1}  u(\sigma_1,\sigma_2). \label{eq:minmax}
\end{equation}

A zero-sum EFG can be represented in \emph{sequence form} \cite{Stengel96}. The sets of sequence-form strategies for players 1 and 2 are $\mathcal{X}$ and $\mathcal{Y}$ respectively. A sequence-form strategy $\vect{x} \in \mathcal{X}$ is a vector indexed by pairs $I \in \mathcal{I}_1$, $a \in A(I)$. The entry $\vect{x}_{(I,a)}$ is the probability of player 1 playing the sequence of actions that reaches $I$ and then playing action $a$.  A special entry, $\vect{x}_\emptyset = 1$, represents the empty sequence.  Any behavioral strategy $\sigma_1 \in \Sigma_1$ has a corresponding sequence-form strategy $\text{SEQ}(\sigma_1)$ where
\begin{equation*}
\text{SEQ}(\sigma_1)_{(I,a)} \coloneqq \pi^{\sigma_1}_1(I) \sigma_1(I,a). \qquad \forall I\in\mathcal{I}_i,a\in A(I)
\end{equation*}

Player $i$ has a unique sequence to reach any history $h\in H$ and, by perfect recall, any information set $I\in\mathcal{I}_i$.  Let $\vect{x}_h$ and $\vect{x}_I$ denote the corresponding entries in $\vect{x}$.  Thus, we are free to write the expected utility as $u(\vect{x}, \vect{y}) = \sum_{z\in Z} \pi_c(z) \vect{x}_{z} \vect{y}_{z} u(z)$.  This is bilinear, i.e., there exists a \emph{payoff matrix} $\mat{A}$ such that $u(\vect{x},\vect{y}) = \vect{x}^\intercal \mat{A}\vect{y}$.  A consequence of perfect recall and the laws of probability is for $I\in\mathcal{I}_1$ that $\vect{x}_I = \sum_{a\in A(I)}\vect{x}_{(I,a)}$ and that $\vect{x}\ge 0$.  These constraints are linear and completely describe the polytope of sequence-form strategies.  Using these together, (\ref{eq:minmax}) can be expressed as a bilinear saddle point problem over the polytopes $\mathcal{X}$ and $\mathcal{Y}$:
\begin{equation}
\max_{\vect{x} \in \mathcal{X}} \min_{\vect{y} \in \mathcal{Y}} \vect{x}^\intercal \mat{A}\vect{y} = \min_{\vect{y} \in \mathcal{Y}} \max_{\vect{x} \in \mathcal{X}} \vect{x}^\intercal \mat{A}\vect{y}
\label{eq:minmaxseq}
\end{equation}

For a convex function $f : \mathcal{X} \rightarrow \mathbb{R}$, let $\nabla f(\vect{x})$ be any element of the subdifferential $\partial f(\vect{x})$, and let $\nabla_{(I,a)} f(\vect{x})$ be the $(I,a)$ element of this subgradient.

\subsection{Counterfactual regret minimization}

\emph{Counterfactual regret minimization} \cite{Zinkevich07} is a large-scale equilibrium-finding algorithm that, in self-play, iteratively updates a strategy profile in a fashion that drives its \emph{counterfactual regret} to zero. This regret is defined in terms of \emph{counterfactual values}. The counterfactual value of reaching information set $I$ is the expected payoff under the counterfactual that the acting player attempts to reach it:
\begin{equation}
v(I,\sigma) = \sum_{h \in I} \pi_{-i}^{\sigma_{-i}}(h) \sum_{z \in Z} \pi^{\sigma}(h,z) u(z) \
\end{equation}
Here $i = P(h)$ for any $h \in I$, and $\pi^{\sigma}(h,z)$ is the probability of reaching $z$ from $h$ under $\sigma$.  Let $\sigma_{I \to a}$ be the profile that plays $a$ at $I$ and otherwise plays according to $\sigma$. For a series of profiles $\sigma^1,..,\sigma^T$, the \emph{average counterfactual regret} of action $a$ at $I$ is $R^T(I,a) = \frac{1}{T} \sum_{t=1}^T v(I,\sigma^t_{I \to a}) - v(I,\sigma^t)$.

To minimize counterfactual regret, CFR employs \emph{regret matching} \cite{Hart00}.  In particular, actions are chosen in proportion to positive regret, $\sigma^{t+1}(I, a) \propto (R^t(I,a))^+$ where $(x)^+=\max(x,0)$.  It follows that the average strategy profile $\overline{\sigma}^T$, defined by $\overline{\sigma}_i^T(I,a) \propto \sum_{t=1}^T \pi_i^{\sigma^t}(I) \sigma^t_i(I,a)$, is an $\mathcal{O}(1/\sqrt{T})$-Nash equilibrium \cite{Zinkevich07}. In sequence form, the average is given by $\vect{\overline{x}^T} = \frac{1}{T} \sum_{t=1}^T \vect{x^t}$.

\section{Solving games with strategy constraints}\label{sec:alg}

We begin by formally introducing the constrained optimization problem for extensive-form games. We specify convex constraints on the set of sequence-form strategies\footnote{Without loss of generality, we assume throughout this paper that the constrained player is player 1, i.e. the maximizing player.} $\mathcal{X}$ with a set of $k$ convex functions $f_i\colon \mathcal{X} \to \mathbb{R}$ where we require $f_i(\vect{x}) \leq 0$ for each $i=1,...,k$. We use constraints on the sequence form instead of on the behavioral strategies because reach probabilities and utilities are linear functions of a sequence-form strategy, but not of a behavioral strategy.

The optimization problem can be stated as:
\begin{align}
\max_{\vect{x} \in \mathcal{X}} \min_{\vect{y} \in \mathcal{Y}}~&\vect{x}^\intercal\mat{A}\vect{y}\label{eq:conopt}\\
&\text{subject to}\qquad f_i(\vect{x}) \leq 0\qquad \text{for}~i=1,...,k\nonumber
\end{align}

The first step toward solving this problem is to incorporate the constraints into the objective with Lagrange multipliers. If the problem is feasible (i.e., there exists a feasible $\vect{x}_f \in \mathcal{X}$ such that $f_i(\vect{x}_f) \leq 0$ for each $i$), then (\ref{eq:conopt}) is equivalent to:
\begin{equation}
\max_{\vect{x} \in \mathcal{X}} \min_{\substack{\vect{y} \in \mathcal{Y}\\\vect{\lambda} \geq \vect{0}}}  \vect{x}^\intercal\mat{A}\vect{y} - \sum_{i=1}^k \lambda_i f_i(\vect{x})\label{eq:lgopt}
\end{equation}

We will now present intuition as to how CFR can be modified to solve (\ref{eq:lgopt}), before presenting the algorithm and proving its convergence.

\subsection{Intuition}

CFR can be seen as doing a saddle point optimization on the objective in (\ref{eq:minmax}), using the gradients\footnote{For a more complete discussion of the connection between CFR and gradient ascent, see \cite{Waugh15}.} of $g(\vect{x},\vect{y}) = \vect{x}^\intercal\mat{A}\vect{y}$ given as
\begin{equation}
\nabla_{\vect{x}} g(\vect{x^t},\vect{y^t}) = \mat{A}\vect{y^t}\qquad \nabla_{\vect{y}} {-g(\vect{x^t},\vect{y^t})} = -(\vect{x^t})^\intercal\mat{A}.
\end{equation}

The intuition behind our modified algorithm is to perform the same updates, but with gradients of the modified utility function
\begin{equation}
h(\vect{x},\vect{y},\vect{\lambda}) = \vect{x}^\intercal\mat{A}\vect{y} - \sum_{i=1}^k \lambda_i f_i(\vect{x}).
\end{equation}
The (sub)gradients we use in the modified CFR update are then
\begin{align}
\begin{split}
&\nabla_{\vect{x}} h(\vect{x^t},\vect{y^t},\vect{\lambda^t}) = \mat{A}\vect{y^t}- \sum_{i=1}^k \lambda^t_i\nabla f_i(\vect{x^t})\\
&\nabla_{\vect{y}} {-h(\vect{x^t},\vect{y^t},\vect{\lambda^t})} = -(\vect{x^t})^\intercal\mat{A}.
\end{split}
\end{align}
Note that this leaves the update of the unconstrained player unchanged. In addition, we must update $\vect{\lambda^t}$ using the gradients $\nabla_{\vect{\lambda}} {-h(\vect{x^t},\vect{y^t},\vect{\lambda^t})} = \sum_{i=1}^k  f_i(\vect{x^t})\vect{e_i}$, which is the $k$-vector with $f_i(\vect{x^t})$ at index $i$. This can be done with any gradient method, e.g. simple gradient ascent with the update rule 
\begin{equation}
\lambda^{t+1}_i = \max(\lambda^t_i + \alpha^t f_i(\vect{x^t}),0)
\end{equation}
for some step size $\alpha^t \propto 1/\sqrt{t}$.

\subsection{Constrained counterfactual regret minimization}

We give the \emph{Constrained CFR} (CCFR) procedure in Algorithm~\ref{alg:ccfr}. The constrained player's strategy is updated with the function CCFR and the unconstrained player's strategy is updated with unmodified CFR. In this instantiation $\vect{\lambda^t}$ is updated with gradient ascent, though any regret minimizing update can be used. We clamp each $\lambda_i^t$ to the interval $[0,\beta]$ for reasons discussed in the following section. Together, these updates form a full iteration of CCFR.

\begin{algorithm}[htb]
\caption{Constrained CFR}
\label{alg:ccfr}
\begin{algorithmic}[1]
\Statex
\Function{CCFR}{$\sigma_i^t$,$\sigma_{-i}^t$,$\vect{\lambda}^t$}
    \For{$I \in \mathcal{I}_i$} \Comment{in reverse topological order}\label{line:loop}
        \For{$a \in A(I)$}
            \State  \begin{varwidth}[t]{\linewidth}
            	$v^t(I,a) \gets \sum_{z \in Z^1[Ia]} \pi^{\sigma_{-i}^t}_{-i}(z)u(z)$\par
            	\hskip\algorithmicindent\hspace{3em} $+ \sum_{I'\in succ(I,a)}\tilde{v}^t(I')$
            	\end{varwidth}
            	\label{line:acfv}\newline
            \hspace*{-\fboxsep}\colorbox{yellow!50}{\parbox{\linewidth-3pt}{%
            \State \begin{varwidth}[t]{\linewidth}
                $\tilde{v}^t(I,a) \gets v^t(I,a)$\par
                \hskip\algorithmicindent\hspace{3em} $- \sum_{i=1}^k \lambda_i^t\nabla_{(I,a)} f_i(\text{SEQ}(\sigma_i^t))$
                \end{varwidth}
                \label{line:tilt}
            }}
        \EndFor
        \State $\tilde{v}^t(I) \gets \sum_{a \in A(I)} \sigma^t_i(I,a)\tilde{v}^t(I,a)$\label{line:cfv}
	\For{$a \in A(I)$}
	   \State $\tilde{r}^t(I,a) \gets \tilde{v}^t(I,a) - \tilde{v}^t(I)$ \label{line:reg1}
	   \State $\tilde{R}^t(I,a) \gets \tilde{R}^{t-1}(I,a) + \tilde{r}^t(I,a)$ \label{line:reg2}
	\EndFor
	\For{$a \in A(I)$}
	   \State $\sigma_i^{t+1}(I,a) \gets \frac{(\tilde{R}^t(I,a))^+}{\sum_{b \in A(I)} (\tilde{R}^t(I,b))^+}$\label{line:rm}
	\EndFor
    \EndFor
    \State \Return $\sigma_i^{t+1}$
\EndFunction
\Statex
\For{$t=1,...,T$}
    \State $\sigma_2^t \gets \Call{CFR}{\sigma_1^{t-1},\sigma_2^{t-1}}$
    \For{$i=1,...,k$}
        \State $\lambda_i^t \gets \lambda_i^{t-1} + \alpha_t f_i(\Psi(\sigma_1^{t-1}))$
        \State $\lambda_i^t \gets \Call{Clamp}{\lambda_i^t,[0,\mathcal{B}]}$
    \EndFor
    \State $\sigma_1^t \gets \Call{CCFR}{\sigma_1^{t-1},\sigma_2^t,\vec{\lambda}^t}$
    \State $\overline{\sigma}_2^t \gets \frac{t-1}{t} \overline{\sigma}_2^{t-1} + \frac{1}{t}\sigma_2^t$
    \State $\overline{\sigma}_1^t \gets \frac{t-1}{t} \overline{\sigma}_1^{t-1} + \frac{1}{t}\sigma_1^t$
\EndFor
\end{algorithmic}
\end{algorithm}

The CCFR update for the constrained player is the same as the CFR update, with the crucial difference of line~\ref{line:tilt}, which incorporates the second part of the gradient $\nabla_{\vect{\vect{x}}}{h}$ into the counterfactual value $v^t(I,a)$. The loop beginning on line~\ref{line:loop} goes through the constrained player's information sets, walking the tree bottom-up from the leafs. The counterfactual value $v^t(I,a)$ is set on line~\ref{line:acfv} using the values of terminal states $Z^1[Ia]$ which directly follow from action $a$ at $I$ (this corresponds to the $\mat{A}\vect{y^t}$ term of the gradient), as well as the already computed values of successor information sets $succ(I,a)$. Line~\ref{line:cfv} computes the value of the current information set using the current strategy. Lines~\ref{line:reg1} and \ref{line:reg2} update the stored regrets for each action. Line~\ref{line:rm} updates the current strategy with regret matching.

\subsection{Theoretical analysis}

In order to ensure that the utilities passed to the regret matching update are bounded, we will require $\vect{\lambda^t}$ to be bounded from above; in particular, we will choose $\vect{\lambda^t} \in [0,\beta]^k$. We can then evaluate the chosen sequence $\vect{\lambda^1},...,\vect{\lambda^T}$ using its regret in comparison to the optimal $\vect{\lambda^*} \in [0,\beta]^k$:
\begin{align}
\begin{split}
R^T_{\lambda}(\beta) \coloneqq \max_{\lambda^* \in [0,\beta]^k} \frac{1}{T} \sum_{t=1}^T \sum_{i=1}^k \bigl[&\lambda_i^* f_i(\text{SEQ}(\sigma_1^t))\\
 &- \lambda_i^t f_i(\text{SEQ}(\sigma_1^t))\bigr].
 \end{split}
\end{align}
We can guarantee $R^T_{\lambda}(\beta) = \mathcal{O}(1/\sqrt{T})$, e.g. by choosing $\vect{\lambda^t}$ with projected gradient ascent~\cite{Zinkevich03}.

We now present the theorems which show that CCFR can be used to approximately solve (\ref{eq:conopt}). In the following thereoms we assume that $T \in \mathbb{N}$, we have some convex, continuous constraint functions $f_1,...,f_k$, and we use some regret-minimizing method to select the vectors $\vect{\lambda}^1,...,\vect{\lambda}^T$ each in $[0,\beta]^k$ for some $\beta \geq 0$.

First, we show that the exploitability of the average strategies approaches the optimal value:
\begin{thm}
If CCFR is used to select the sequence of strategies $\sigma_1^1,...,\sigma_1^T$ and CFR is used to select the sequence of strategies $\sigma_2^1,...,\sigma_2^T$, then the following holds:
\begin{align}
&\max_{\substack{\sigma_1^* \in \Sigma_1\\\text{s.t. }f_i(\sigma_1^*) \leq 0~\forall i}} u(\sigma_1^*,\overline{\sigma}_2^T) - \min_{\sigma_2^* \in \Sigma_2} u(\overline{\sigma}_1^T,\sigma_2^*)\nonumber\\
&\hspace{5em}\leq \frac{4\left(\Delta_u + k\beta F\right)M \sqrt{|A|}}{\sqrt{T}} + 2R_\lambda^T(\beta)
\label{eq:expbound}
\end{align}
where $\Delta_u = \max_z u(z) - \min_z u(z)$ is the range of possible utilities, $|A|$ is the maximum number of actions at any information set, $k$ is the number of constraints, $F=\max_{\vect{x},i}||\nabla f_i(\vect{x})||_1$ is a bound on the subgradients\footnote{Such a bound must exist as the strategy sets are compact and the constraint functions are continuous.}, and $M$ is a game-specific constant.
\label{thm:ccfrex}
\end{thm}
All proofs are given in the appendix. Theorem~\ref{thm:ccfrex} guarantees that the constrained exploitability of the final CCFR strategy profile converges to the minimum exploitability possible over the set of feasible profiles, at a rate of $\mathcal{O}(1/\sqrt{T})$ (assuming a suitable regret minimizer is used to select $\vect{\lambda}^t$).

In order to establish that CCFR approximately solves optimization (\ref{eq:conopt}), we must also show that the CCFR strategies converge to being feasible.  In the case of arbitrary $\beta \geq 0$:
\begin{thm}
If CCFR is used to select the sequence of strategies $\sigma_1^1,...,\sigma_1^T$ and CFR is used to select the sequence of strategies $\sigma_2^1,...,\sigma_2^T$, then the following holds:
\begin{align}
&f_i(\text{SEQ}(\overline{\sigma}_1^T)) \leq \frac{R_\lambda^T(\beta)}{\beta} + \frac{\left(\Delta_u + 2k\beta F\right)M \sqrt{|A|}}{\beta\sqrt{T}} + \frac{\Delta_u}{\beta}\nonumber\\
&\hspace{15em}\forall i \in \{1,...,k\}\label{eq:conbound}
\end{align}
\label{thm:conbound}
\end{thm}
This theorem guarantees that the CCFR strategy converges to the feasible set at a rate of $\mathcal{O}(1/\sqrt{T})$, up to an approximation error of $\Delta_u/\beta$ induced by the bounding of $\vect{\lambda^t}$.

We can eliminate the approximation error when $\beta$ is chosen large enough for some optimal $\vect{\lambda}^*$ to lie within the bounded set $[0,\beta]^k$. In order to establish the existence of such a $\vect{\lambda}^*$, we must assume a constraint qualification such as Slater's condition, which requires the existance of a feasible $\vect{x}$ which strictly satisfies any nonlinear constraints ($f_i(\vect{x}) \leq 0$ for all $i$ and $f_j(\vect{x}) < 0$ for all nonlinear $f_j$). Then there exists a finite $\vect{\lambda}^*$ which is a solution to optimization (\ref{eq:lgopt}), which we can use to give the bound:
\begin{thm}
Assume that $f_1,...,f_k$ satisfy a constraint qualification such as Slater's condition, and define $\vect{\lambda}^*$ to a finite solution for $\vect{\lambda}$ in the resulting optimization (\ref{eq:lgopt}). Then if $\beta$ is chosen such that $\beta > \lambda_i^*$ for all $i$, and CCFR and CFR are used to respectively select the strategy sequences $\sigma_1^1,...,\sigma_1^T$ and $\sigma_2^1,...,\sigma_2^T$, the following holds:
\begin{align}
&f_i(\text{SEQ}(\overline{\sigma}_1^T)) \leq \frac{R_\lambda^T(\beta)}{\beta-\lambda^*_i} + \frac{2\left(\Delta_u + k\beta F\right)M \sqrt{|A|}}{(\beta-\lambda^*_i)\sqrt{T}}\nonumber\\
&\hspace{15em}\forall i \in \{1,...,k\}\label{eq:conboundopt}
\end{align}
\label{thm:conboundopt}
\end{thm}
In this case, the CCFR strategy converges fully to the feasible set, at a rate of $\mathcal{O}(1/\sqrt{T})$, given a suitable choice of regret minimizer for $\vect{\lambda}^t$. We provide an explicit example of such a minimizer in the following corollary:
\begin{cor}
If the conditions of Theorem~\ref{thm:conboundopt} hold and, in addition, the sequence $\vect{\lambda}^1,...,\vect{\lambda}^T$ is chosen using projected gradient descent with constant learning rate $\alpha^t = \beta/(G\sqrt{T})$ where $G = \max_{i,\vect{x}} f_i(\vect{x})$, then the following hold:
\begin{align}
&f_i(\text{SEQ}(\overline{\sigma}_1^T)) \leq \frac{\beta G + 2\left(\Delta_u + k\beta F\right)M \sqrt{|A|}}{(\beta-\lambda^*_i)\sqrt{T}}\nonumber\\
&\hspace{15em}\forall i \in \{1,...,k\}
\end{align}
\begin{align}
&\max_{\substack{\sigma_1^* \in \Sigma_1\\\text{s.t. }f_i(\sigma_1^*) \leq 0~\forall i}} u(\sigma_1^*,\overline{\sigma}_2^T) - \min_{\sigma_2^* \in \Sigma_2} u(\overline{\sigma}_1^T,\sigma_2^*)\nonumber\\
&\hspace{5em}\leq \frac{4\left(\Delta_u + k\beta F\right)M \sqrt{|A|} + 2\beta G}{\sqrt{T}}
\end{align}
\end{cor}

\begin{proof}
This follows from using the projected gradient descent regret bound \cite{Zinkevich03} to give\\ $R_\lambda^T(\beta) \leq \frac{\beta^2}{2\alpha^T} + \frac{G^2}{2}\sum_{t=1}^T\alpha^t \leq \frac{\beta G}{\sqrt{T}}$.
\end{proof}

Finally, we discuss how to choose $\beta$. When there is a minimum acceptable constrant violation, $\beta$ can be selected with Theorem~\ref{thm:conbound} to guarantee that the violation is no more than the specified value, either asymptotically or after a specified number of iterations $T$. When no amount of constraint violation is acceptable, $\beta$ should be chosen such that $\beta \geq \lambda_i^*$ by Theorem~\ref{thm:conboundopt}. If $\vect{\lambda}^*$ is unknown, CCFR can be run with an arbitrary $\beta$ for a number of iterations. If the average $\frac{1}{T}\sum_{t=1}^T\lambda_i^t$ is close to $\beta$, then $\beta \leq \lambda^*_i$, so $\beta$ is doubled and CCFR run again. Otherwise, it is guaranteed that $\beta > \lambda^*_i$ and CCFR will converge to a solution with no constraint violation.

\section{Related Work}

To the best of our knowledge, no previous work has proposed a technique for solving either of the optimizations (\ref{eq:conopt}) or (\ref{eq:lgopt}) for general constraints in extensive-form games. Optimization (\ref{eq:lgopt}) belongs to a general class of saddle point optimizations for which a number of accelerated methods with $\mathcal{O}(1/T)$ convergence have been proposed \cite{Nemirovski04,Nesterov05a,Nesterov05b,Juditsky11,Chambolle11}. These methods have been applied to unconstrained equilibrium computation in extensive-form games using a family of prox functions initially proposed by Hoda et. al. \cite{Hoda10,Kroer15,Kroer17}. Like CFR, these algorithms could be extended to solve the optimization (\ref{eq:lgopt}).

Despite a worse theoretical dependence on $T$, CFR is preferred to accelerated methods as our base algorithm for a number of practical reasons.
\begin{itemize}
\item CFR can be easily modified with a number of different sampling schemes, adapting to sparsity and achieving greatly improved convergence over the deterministic version \cite{Lanctot09}. Although the stochastic mirror prox algorithm has been used to combine an accelerated update with sampling in extensive-form games, each of its iterations still requires walking each player's full strategy space to compute the prox functions, and it has poor performance in practice \cite{Kroer15}.
\item CFR has good empirical performance in imperfect recall games \cite{Waugh09b} and even provably converges to an equilibrium in certain subclasses of well-formed games \cite{Lanctot12,Lisy16}, which we will make use of in Section~\ref{sec:tg}. The prox function used by the accelerated methods is ill-defined in all imperfect recall games.
\item CFR theoretically scales better with game size than do the accelerated techniques. The constant $M$ in the bounds of Theorems~\ref{thm:ccfrex}-\ref{thm:conboundopt} is at worst $|\mathcal{I}|$, and for many games of interest is closer to $|\mathcal{I}|^{1/2}$ (\citeauthor{Burch17}~\citeyear{Burch17}, Section~3.2). The best convergence bound for an accelerated method depends in the worst case on $|\mathcal{I}|^2 2^d$ where $d$ is the depth of the game tree, and is at best $|\mathcal{I}|2^d$ \cite{Kroer17}.
\item The CFR update can be modified to CFR+ to give a guaranteed bound on tracking regret and greatly improve empirical performance \cite{Tammelin15}. CFR+ has been shown to converge with initial rate faster than $\mathcal{O}(1/T)$ in a variety of games (\citeauthor{Burch17}~\citeyear{Burch17}, Sections~4.3-4.4).
\item Finally, CFR is not inherently limited to $\mathcal{O}(1/\sqrt{T})$ worst-case convergence. Regret minimization algorithms can be optimistically modified to give $\mathcal{O}(1/T)$ convergence in self-play \cite{Rakhlin13}. Such a modification has been applied to CFR (\citeauthor{Burch17}~\citeyear{Burch17}, Section~4.4).
\end{itemize}

We describe CCFR as an extension of deterministic CFR for ease of exposition. All of the CFR modifications described in this section can be applied to CCFR out-of-the-box.

\section{Experimental evaluation}

We present two domains for experimental evaluation in this paper. In the first, we use constraints to model a secondary objective when generating strategies in a model security game. In the second domain, we use constraints for opponent modeling in a small poker game. We demonstrate that using constraints for modeling data allows us to learn counter-strategies that approach optimal counter-strategies as the amount of data increases. Unlike previous opponent modeling techniques for poker, we do not require our data to contain full observations of the opponent's private cards for this guarantee to hold.

\subsection{Transit game}
\label{sec:tg}

\begin{figure}[t]
\centering
\tgScheme
\caption{Transit Game}
\label{fig:tg}
\end{figure}

The \emph{transit game} is a model security game introduced in \cite{Bosansky15}. With size parameter $w$, the game is played on an 8-connected grid of size $2w \times w$ (see Figure~\ref{fig:tg}) over $d=2w+4$ time steps. One player, the \emph{evader}, wishes to cross the grid from left to right while avoiding the other player, the \emph{patroller}. Actions are movements along the edges of the grid, but each move has a probability $0.1$ of failing. The evader receives $-1$ utils for each time he encounters the patroller, $1$ utils when he escapes on reaching the east end of the grid, and $-0.02$ utils for each time step that passes without escaping. The patroller receives the negative of the evader's utils, making the game zero-sum. The players observe only their own actions and locations.

The patroller has a secondary objective of minimizing the risk that it fails to return to its base ($s^0_p$ in Figure~\ref{fig:tg}) by the end of the game. In the original formulation, this was modeled using a large utility penalty when the patroller doesn't end the game at its base. For the reasons discussed in the introduction, it is more natural to model this objective as a linear constraint on the patroller's strategy, bounding the maximum probability that it doesn't return to base.

For our experiments, we implemented CCFR on top of the NFGSS-CFR algorithm described in \cite{Lisy16}. In the NFGSS framework, each information set is defined by only the current grid state and the time step; history is not remembered. This is a case of imperfect recall, but our theory still holds as the game is well-formed. The constraint on the patroller is defined as 
\begin{equation*}
\sum_{s^d,a} \pi^{\sigma_p}(s^d)\sigma(s^d,a) \sum_{s^{d+1} \neq s^0_p} T(s^d,a,s^{d+1}) \leq b_r
\end{equation*}
 where $s^d,a$ are state action pairs at time step $d$, $T(s^d,a,s^{d+1})$ is the probability that $s^{d+1}$ is the next state given that action $a$ is taken from $s^d$, and $b_r$ is the chosen risk bound. This is a well-defined linear constraint despite imperfect recall, as $\pi^{\sigma_p}(s^d)$ is a linear combination over the sequences that reach $s^d$. We update the CCFR constraint weights $\lambda$ using stochastic gradient ascent with constant step size $\alpha^t=1$, which we found to work well across a variety of game sizes and risk bounds. In practice, we found that bounding $\vect{\lambda}$ was unnecessary for convergence.
 
Previous work has shown that double oracle (DO) techniques outperform solving the full game linear program (LP) in the unconstrained transit game \cite{Bosansky15,Lisy16}. However, an efficient best response oracle exists in the unconstrained setting only because a best response is guaranteed to exist in the space of pure strategies, which can be efficiently searched. Conversely, constrained best responses might exist only in the space of mixed strategies, meaning that the best response computation requires solving an LP of comparable size to the LP for the full game Nash equilibrium. This makes DO methods inappropriate for the general constrained setting, so we omit comparison to DO methods in this work.

\paragraph{Results}

\begin{figure*}[t]
\centering
\begin{subfigure}[t]{0.33\linewidth}
\centering
  \includegraphics[width=\linewidth]{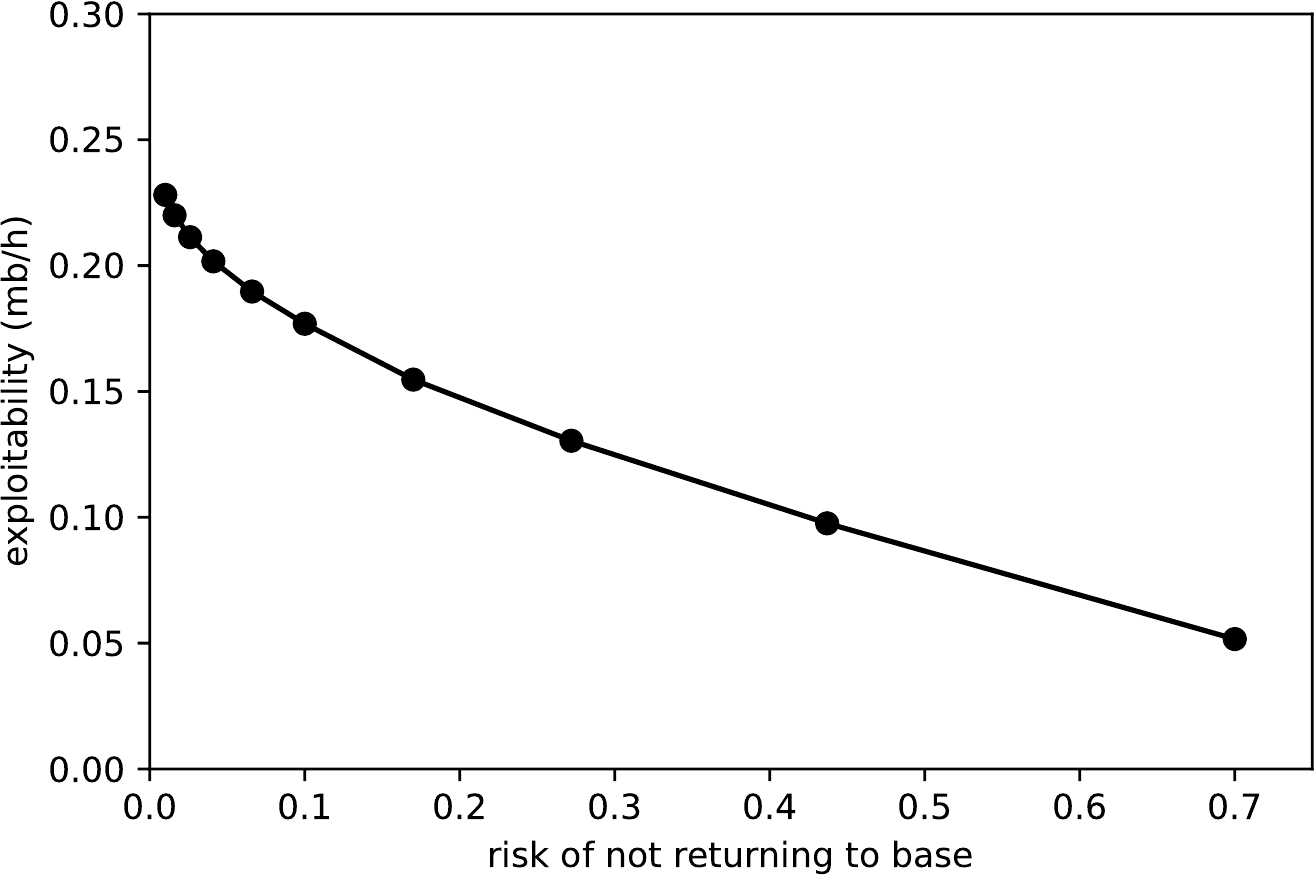}
  \caption{Risk vs Exploitability}
  \label{fig:tgRisk}
\end{subfigure}%
\begin{subfigure}[t]{0.33\linewidth}
\centering
  \includegraphics[width=\linewidth]{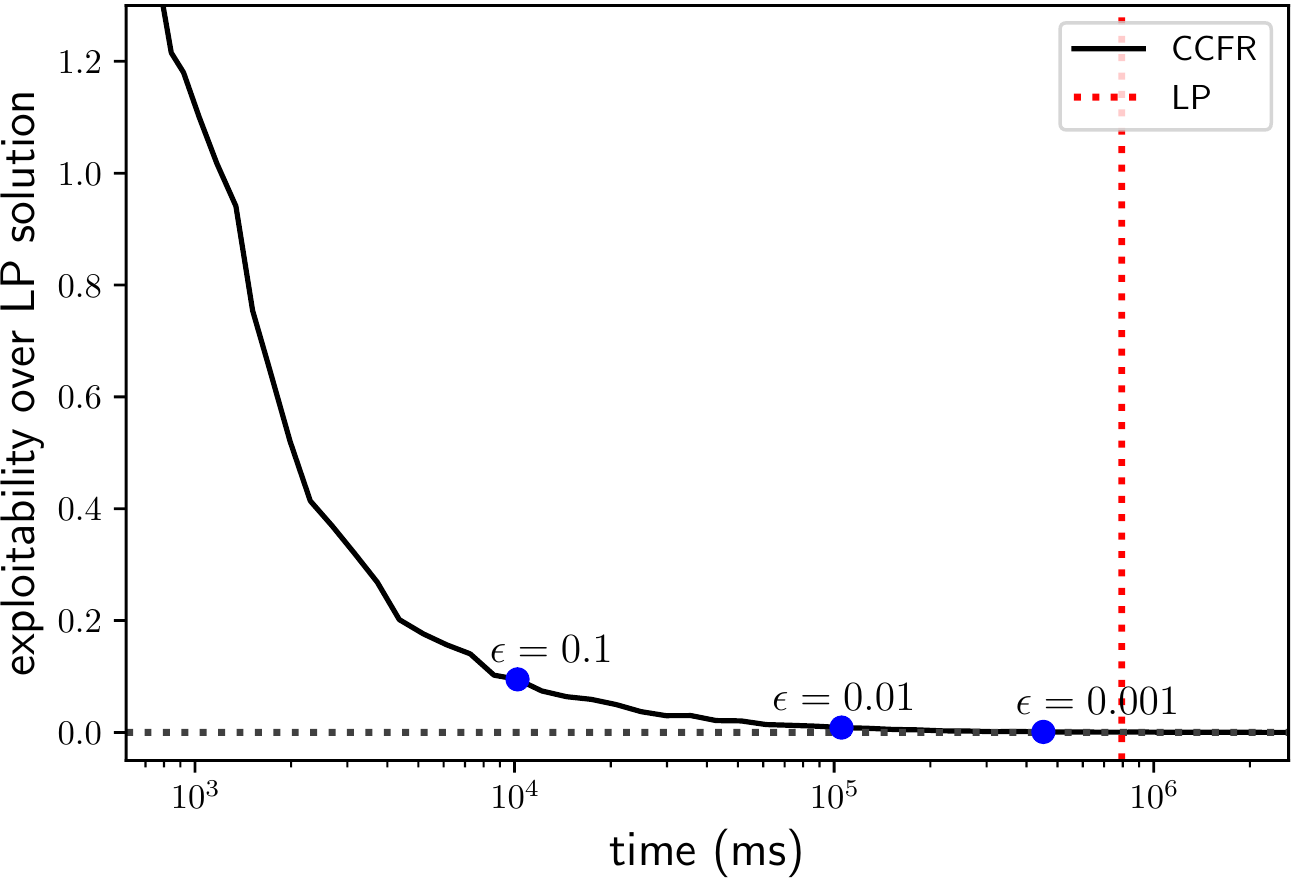}
  \caption{CCFR convergence}
  \label{fig:tgConv}
\end{subfigure}%
\begin{subfigure}[t]{0.33\linewidth}
\centering
  \includegraphics[width=\linewidth]{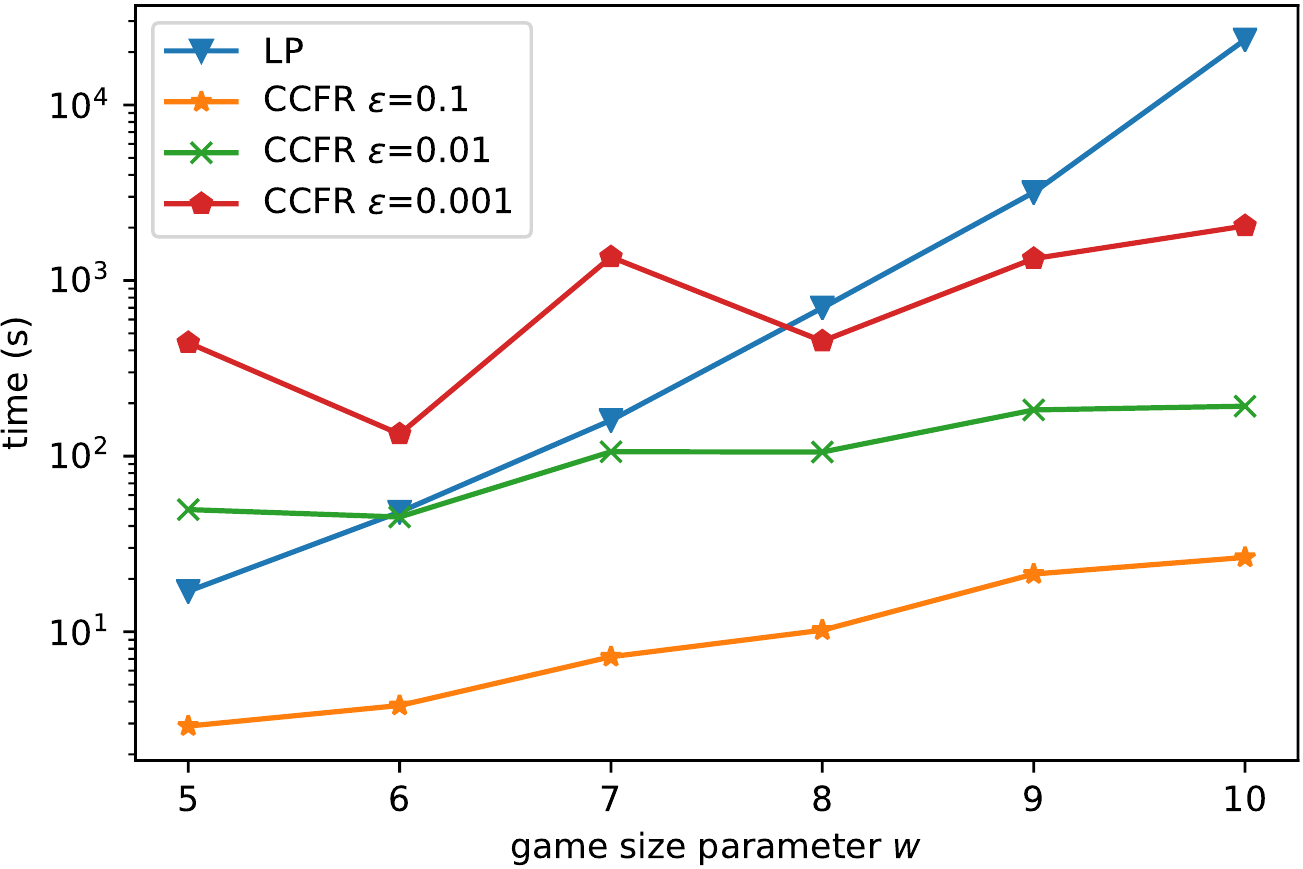}
  \caption{Convergence time}
  \label{fig:tgTime}
\end{subfigure}
\caption{Risk and exploitability of final CCFR strategies with game size $w=8$ (a), convergence of CCFR exploitability to LP exploitability with $w=8$ and risk bound $b_r = 0.1$ (b), and dependence of convergence time on game size for LP and CCFR methods (c). All algorithms are deterministic, so times are exact.}
\end{figure*}
We first empirically demonstrate that CCFR converges to optimality by comparing its produced strategies with strategies produced by solving the LP representation of the game with the simplex solver in IBM ILOG CPLEX 12.7.1. Figure~\ref{fig:tgRisk} shows the risk and exploitability for strategies produced by running CCFR for 100,000 iterations on a game of size $w=8$, with a variety of values for the risk bound $b_r$. In each case, the computed strategy had risk within $0.001$ of the specified bound $b_r$, and exploitability within $0.001$ of the corresponding LP strategy (not shown because the points are indistinguishable). The convergence over time for one particular case, $b_r=0.1$, is shown in Figure~\ref{fig:tgConv}, where the plotted value is the difference in exploitability between the average CCFR strategy and the LP strategy, shown with a log-linear scale. The vertical line shows the time used to compute the LP strategy.

Convergence times for the CPLEX LP and CCFR with risk bound $b_r=0.1$ are shown on a log scale for a variety of game sizes $w$ in Figure~\ref{fig:tgTime}. The time for CCFR is presented for a variety of precisions $\epsilon$, which bounds both the optimality of the final exploitability and the violation of the risk bound. The points for game size $w=8$ are also shown in Figure~\ref{fig:tgConv}. The LP time is calculated with default precision $\epsilon=10^{-6}$. Changing the precision to a higher value actually results in a slower computation, due to the precision also controlling the size of allowed infeasibility in the Harris ratio test \cite{Klotz13}.

Even at $w=6$, a game which has relatively small strategy sizes of $\sim$6,000 values, CCFR can give a significant speedup for a small tradeoff in precision. At $w=8$ and larger, the LP is clearly slower than CCFR even for the relatively modest precision of $\epsilon=0.001$. For game size $w=10$, with strategy sizes of $\sim$25,000 values, the LP is more than an order of magnitude slower than high precision CCFR.

\subsection{Opponent modeling in poker}

In multi-agent settings, strategy constraints can serve an additional purpose beyond encoding secondary objectives. Often, when creating a strategy for one agent, we have partial information on how the other agent(s) behave. A way to make use of this information is to solve the game with constraints on the other agents' strategies, enforcing that their strategy in the solution is consistent with their observed behavior. As a motivating example, we consider poker games in which we always observe our opponent's actions, but not necessarily the private card(s) that they hold when making the action.

In poker games, if either player takes the \emph{fold} action, the other player automatically wins the game. Because the players' private cards are irrelevant to the game outcome in this case, they are typically not revealed. We thus consider the problem of opponent modeling from observing past games, in which the opponent's \emph{hand} of private card(s) is only revealed when a \emph{showdown} is reached and the player with the better hand wins. Most previous work in opponent modeling has either assumed full observation of private cards after a fold \cite{Johanson07,Johanson09} or has ignored observations of opponent actions entirely, instead only using observed utilities \cite{Bard13}. The only previous work which uses these partial observations has no theoretical guarantees on solution quality \cite{Ganzfried11}.

We first collect data by playing against the opponent with a probe strategy, which is a uniformly random distribution over the non-fold actions. To model the opponent in an unbiased way, we generate two types of sequence-form constraints from this data. First, for each possible sequence of public actions and for each of our own private hands, we build an unbiased confidence interval on the probability that we are dealt the hand and the public sequence occurs. This probability is a weighted sum of opponent sequence probabilities over their possible private cards, and thus the confidence bounds become linear sequence-form constraints. Second, for each terminal history that is a showdown, we build a confidence interval on the probability that the showdown is reached. In combination, these two sets of constraints guarantee that the CCFR strategy converges to a best response to the opponent strategy as the number of observed games increases. A proof of convergence to a best response and full details of the constraints are provided in the appendix.

\paragraph{Infeasible constraints}

Because we construct each constraint separately, there is no guarantee that the full constraint set is simultaneously feasible. In fact, in our experiments the constraints were typically mildly infeasible. However, this is not a problem for CCFR, which doesn't require feasible constraints to have well-defined updates. In fact, because we bound the Lagrange multipliers, CCFR still theoretically converges to a sensible solution, especially when the total infeasibility is small. For more details on how CCFR handles infeasibility, see the appendix.

\paragraph{Results}

We ran our experiments in Leduc Hold'em \cite{Southey05}, a small poker game played with a six card deck over two betting rounds. To generate a target strategy profile to model, we solved the "JQ.K/pair.nopair" abstracted version of the game \cite{Waugh09}. We then played a probe strategy profile against the target profile to generate constraints as described above, and ran CCFR twice to find each half of a counter-profile that is optimal against the set of constrained profiles. We used gradient ascent with step size $\alpha^t = 1000/\sqrt{t}$ to update the $\vect{\lambda}$ values, and ran CCFR for $10^6$ iterations, which we found to be sufficient for approximate convegence with $\varepsilon < 0.001$.

\begin{figure}[th]
\centering
  \includegraphics[width=0.85\linewidth]{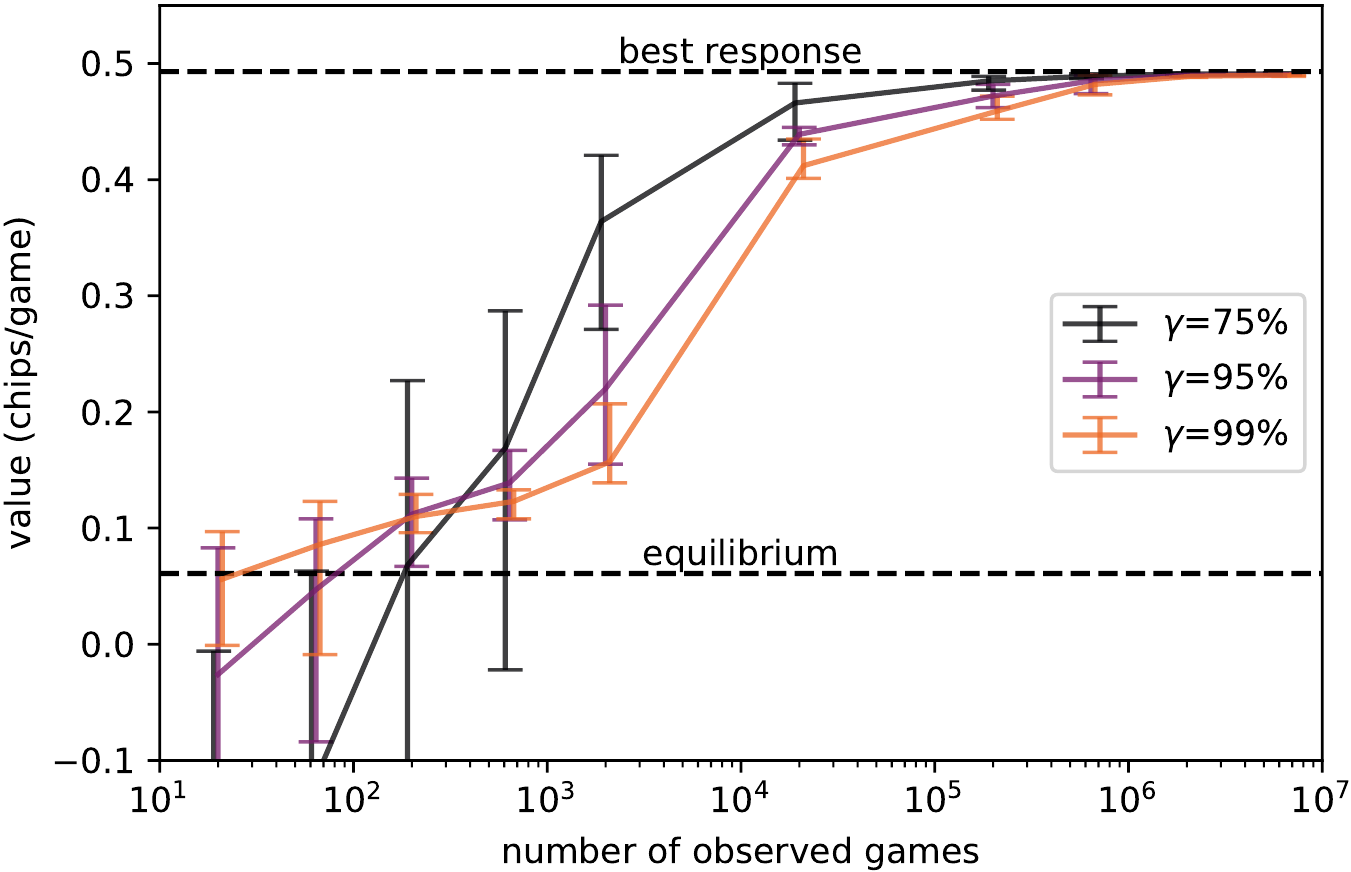}
  \caption{Performance of counter-profiles trained with CCFR against a target profile in Leduc Hold'em. Results are averaged over 10 runs, with minimum and maximum values shown as error bars. Values achieved by a Nash equilibrium profile and a best response profile are shown as horizontal lines for reference.}
  \label{fig:leduc}
\end{figure}

We evaluate how well the trained counter-profile performs when played against the target profile, and in particular investigate how this performance depends on the number of games we observe to produce the counter-profile, and on the confidence $\gamma$ used for the confidence interval constraints. Results are shown in Figure~\ref{fig:leduc}, with a log-linear scale. With a high confidence $\gamma=99\%$ (looser constraints), we obtain an expected value that is better than the equilibrium expected value with fewer than 100 observed games on average, and with fewer than 200 observed games consistently. Lower confidence levels (tighter constraints) resulted in more variable performance and poor average value with small numbers of observed games,  but also faster learning as the number of observed games increased.  For all confidence levels, the expected value converges to the best response value as the number of observed games increases. 

\section{Conclusion}

Strategy constraints are a powerful modeling tool in extensive-form games. Prior to this work, solving games with strategy constraints required solving a linear program, which scaled poorly to many of the very large games of practical interest. We introduced CCFR, the first efficient large-scale algorithm for solving extensive-form games with general strategy constraints. We demonstrated that CCFR is effective at solving sequential security games with bounds on acceptable risk. We also introduced a method of generating strategy constraints from partial observations of poker games, resulting in the first opponent modeling technique that has theoretical guarantees with partial observations. We demonstrated the effectiveness of this technique for opponent modeling in Leduc Hold'em.

\section{Acknowledgments}

The transit game experiments were implemented with code made publically available by the game theory group of the Artificial Intelligence Center at Czech Technical University in Prague. This research was supported by Alberta Innovates, Alberta Advanced Education, and the Alberta Machine Intelligence Institute (Amii). Computing resources were provided by Compute Canada and Calcul Qu{\'e}bec.

\fontsize{9pt}{10pt}
\selectfont
\bibliographystyle{aaai}
\bibliography{constrainedcfr}

\onecolumn
\normalsize
\appendix

\setcounter{thm}{3}

\section{Notation}

Define $\Psi(\sigma_i)=\text{SEQ}(\sigma_i)$ for brevity. Let $\tilde{u}$ be the modified utility function:
\begin{equation}
\tilde{u}(\sigma_1,\sigma_2,\lambda) = u(\sigma_1,\sigma_2) - \sum_{i=1}^k \lambda_i f_i(\Psi(\sigma_1))
\end{equation}

For any history $h \in H$, define
\begin{equation}
Z[h] = \{ z \in Z \mid h \sqsubset z \}
\end{equation}
to be the set of terminal histories reachable from $h$. In addition, define
\begin{equation}
Z^1[h] = \{ z \in Z \mid h \sqsubset z \text{ and} (\nexists h' \in H_1) h \sqsubset h' \sqsubset z \}
\end{equation}
to be the subset of those histories which include no player 1 actions after the action at $h$. Define
\begin{equation}
Succ(h,a) = \{ h' \in H_1 \mid h \sqsubset h' \text{ and} (\nexists h'' \in H_1) h \sqsubset h'' \sqsubset h' \}
\end{equation}
to be set of histories where player 1 might act next after taking action $a$ at $H$. Also, for any $h \in H$ define
\begin{equation}
D(h) = \{ h' \in H_1 \mid h \sqsubseteq h' \}
\end{equation}
to be the set of all histories where player 1 might act after $h$. Note that if $P(h) = 1$, then $h \in D(h)$.

We extend all of the preceding definitions to information sets in the natural way:
\begin{equation}
\begin{split}
Z[I] &= \bigcup_{h \in I} Z[h] \\
Z^1[I] &= \bigcup_{h \in I} Z^1[h]
\end{split}
\qquad
\begin{split}
Succ(I,a) &= \{ I' \in \mathcal{I}_1 \mid (\exists h' \in I') h' \in \bigcup_{h \in I} Succ(h,a) \}\\ 
D(I) &= \{ I' \in \mathcal{I}_1 \mid (\exists h' \in I') h' \in \bigcup_{h \in I} D(h) \}
\end{split}
\end{equation}

Consider two information sets $I,I' \in \mathcal{I}_1$ such that $I' \in D(I)$. By perfect recall, the series of actions that player 1 takes from $I$ to $I'$ must be unique, no matter which $h \in I$ is the starting state. Thus we may define $\pi_1^{\sigma}(I,I') = \pi_1^{\sigma}(h,h')$ for any $h \in I$ and $h' \in I'$ where $h \sqsubseteq h'$. Similarly, for $z \in Z$, we can define $\pi^{\sigma}(I,z) = \pi^{\sigma}(h,z)$ where $h$ is the unique $h \in I$ such that $h \sqsubset z$ (or $\pi^{\sigma}(I,z) = 0$ if no such $h$ exists).

For $a \in A(I)$, we will use $v(I,\sigma,a)$ as short for $v(I,\sigma|_{I \to a})$, where $\sigma|_{I \to a}$ is the strategy profile that always takes action $a$ at $I$, but everywhere else is identical to $\sigma$.

\section{Proof of Theorems}

For any information set $I \in \mathcal{I}_1$, action $a \in A(I)$, and $t \in \{1,...,T\}$, let $c^t(I,a)$ be the constraint tilt added to $v^t(I,a)$, i.e.
\begin{equation}
c^t(I,a) = \sum_{i=1}^k \lambda^t_i \nabla_{(I,a)}f_i(\Psi(\sigma^t_1))
\end{equation}

Given a $t \in \{1,...,T\}$, define the tilted counterfactual value $\tilde{v}^t$ for strategy profile $\sigma$, information set $I \in \mathcal{I}_1$, and action $a \in A(I)$ by the recurrences
\begin{align}
\tilde{v}^t(I,\sigma,a) &= \sum_{z \in Z^1[Ia]} \pi^{\sigma_{2}}_{-1}(z)u(z) + \sum_{I'\in Succ(I,a)}\tilde{v}^t(I',\sigma) - c^t(I,a)\label{tcfva-recur}\\
\tilde{v}^t(I,\sigma) &= \sum_{a \in A(I)} \sigma_1(I,a) \tilde{v}^t(I,\sigma,a)\label{tcfv-recur}
\end{align}

\begin{lem}
Given a strategy profile $\sigma$ and a $t \in \{1,...,T\}$, if $\tilde{v}^t$ is defined by recurrences (\ref{tcfva-recur}) and (\ref{tcfv-recur}) for every $I \in \mathcal{I}_1$ and $a \in A(I)$, then for any such $I,a$ the following hold:
\begin{align}
\tilde{v}^t(I,\sigma,a) &= v(I,\sigma,a) - c^t(I,a) - \sum_{\substack{I' \in D(I)\\Ia \sqsubseteq I'}} \sum_{a' \in A(I')} \pi_1^{\sigma_1}(Ia,I'a') c^t(I',a')\label{tcfva-closed}\\
\tilde{v}^t(I,\sigma) &= v(I,\sigma) - \sum_{I' \in D(I)} \sum_{a \in A(I')} \pi_1^{\sigma_1}(I,I'a) c^t(I',a)\label{tcfv-closed}
\end{align}
\label{lemcfv}
\end{lem}

\begin{proof}
We proceed by strong induction on $|D(I)|$. In the base case $D(I) = \{I\}$, we have that $Succ(I,a)$ is empty for each $a \in A(I)$, that there is no $I' \in D(I)$ such that $Ia \sqsubseteq I'$, and that $Z[I] = Z^1[I]$. Thus:
\begin{align*}
\tilde{v}^t(I,\sigma,a) &= \sum_{z \in Z^1[Ia]} \pi^{\sigma_{2}}_{-1}(z)u(z) + \sum_{I'\in Succ(I,a)}\tilde{v}^{\sigma}_t(I') - c^t(I,a)\\
&= \sum_{z \in Z[Ia]} \pi^{\sigma_{2}}_{-1}(z)u(z) - c^t(I,a)\\
&= v(I,\sigma,a) - c^t(I,a) - \sum_{\substack{I' \in D(I)\\Ia \sqsubseteq I'}} \sum_{a' \in A(I')} \pi_1^{\sigma_1}(Ia,I'a') c^t(I',a')
\end{align*}
and
\begin{align*}
\tilde{v}^t(I,\sigma) &= \sum_{a \in A(I)} \sigma_1(I,a) \tilde{v}^t(I,\sigma,a)\\
&= \sum_{a \in A(I)} \sigma_1(I,a) \left[ \sum_{z \in Z[Ia]} \pi^{\sigma_{2}}_{-1}(z)u(z) - c^t(I,a) \right]\\
&= \sum_{z \in Z[I]} \pi^{\sigma_1}_1(I,z) \pi^{\sigma_{2}}_{-1}(z)u(z) - \sum_{a \in A(I)} \pi_1^{\sigma_1}(I,Ia) c^t(I,a)\\
&= v(I,\sigma) - \sum_{I' \in D(I)} \sum_{a \in A(I')} \pi_1^{\sigma_1}(I,I'a) c^t(I',a)
\end{align*}

For the inductive step, we consider some $I \in \mathcal{I}_1$ and assume that the lemma holds for all $I'$ such that $|D(I')| < |D(I)|$. In particular, note that if $I' \in succ(I)$, then $D(I') \subset D(I)$, but $I \in D(I)$ and by perfect recall $I \notin D(I')$, so we can assume the lemma holds for all $I' \in succ(I)$.
\begin{align*}
\tilde{v}^t(I,\sigma,a) &= \sum_{z \in Z^1[Ia]} \pi^{\sigma_{-1}}_{2}(z)u(z) + \sum_{I'\in Succ(I,a)}\tilde{v}^t(I',\sigma) - c^t(I,a)\\
&= \sum_{z \in Z^1[Ia]} \pi_1^{\sigma_1}(Ia,z) \pi_{-1}^{\sigma_2}(z)u(z) - \sum_{I' \in Succ(I,a)} \left[ v^t(I',\sigma) - \sum_{I'' \in D(I')} \sum_{a' \in A(I'')} \pi_1^{\sigma}(I',I''a') c^t(I'',a') \right] - c^t(I,a)\\
&= \sum_{z \in Z^1[Ia]} \pi_1^{\sigma_1}(Ia,z) \pi_{-1}^{\sigma_2}(z)u(z) + \sum_{I' \in Succ(I,a)} \sum_{z \in Z[I']} \pi_1^{\sigma_1}(Ia,z)\pi_{-1}^{\sigma_2}(z)u(z)\\
&\hspace{20pt} - \sum_{I' \in Succ(I,a)} \sum_{I'' \in D(I')}\sum_{a \in A(I'')} \left[\pi_1^{\sigma_1}(Ia,I''a) c^t(I'',a)\right] - c^t(I,a) \\
&= \sum_{z \in Z[Ia]} \pi_1^{\sigma_1}(Ia,z) \pi_{-1}^{\sigma_2}(z)u(z) - \sum_{\substack{I' \in D(I)\\Ia \sqsubseteq I'}} \sum_{a \in A(I')} \left[\pi_1^{\sigma_1}(Ia,I'a) c^t(I',a)\right] - c^t(I,a)\\
\end{align*}
In the last step, we make use of the fact that
\begin{equation} 
Z[I] = Z^1[I] \cup \bigcup_{I' \in Succ(I)} Z[I'] \text{\qquad and\qquad} \{I' \in D(I) \mid Ia \sqsubseteq I' \} = \bigcup_{I' \in Succ(I,a)} D(I')
\end{equation}
and perfect recall requires that each of these component sets are disjoint.

Substituting $v(I,\sigma,a)$ for its definition completes the inductive step for $\tilde{v}^t(I,\sigma,a)$. Then we also have
\begin{align*}
\tilde{v}^t(I,\sigma) &= \sum_{a \in A(I)} \sigma_1(I,a) \tilde{v}^t(I,\sigma,a)\\
&= \sum_{a \in A(I)} \sigma_1(I,a) \left[ \sum_{z \in Z[Ia]} \pi_1^{\sigma_1}(Ia,z) \pi_{-1}^{\sigma_2}(z)u(z) - \sum_{\substack{I' \in D(I)\\Ia \sqsubseteq I'}} \sum_{a' \in A(I')} \left[\pi_1^{\sigma_1}(Ia,I'a') c^t(I',a')\right] - c^t(I,a) \right]\\
&= \sum_{z \in Z[I]} \pi_1^{\sigma_1}(I,z) \pi_{-1}^{\sigma_2}(z)u(z) - \sum_{\substack{I' \in D(I)\\I' \neq I}} \sum_{a \in A(I')} \left[\pi_1^{\sigma_1}(I,I'a) c^t(I',a)\right] - \sum_{a \in A(I)} \sigma_1(I,a) c^t(I,a)\\
&= v^{\sigma}(I) - \sum_{I' \in D(I)} \sum_{a \in A(I')} \pi_1^{\sigma_1}(I,I'a) c^t(I',a)\label{tcfv-closed}
\end{align*}

\end{proof}

Next we analyze the tilted regrets that are used in CCFR. Note that $\tilde{v}^t(I,a)$ and $\tilde{v}^t(I)$ are defined in the algorithm with recurrences  (\ref{tcfva-recur}) and (\ref{tcfv-recur}) where $\sigma = \sigma^t$. Then we also have
\begin{equation}
\tilde{r}^t(I,a) = \tilde{v}^t(I,\sigma^t,a) - \tilde{v}^t(I,\sigma^t)
\end{equation}
and for any $\sigma_1' \in \Sigma_1$ define
\begin{equation}
\tilde{r}^t(I,\sigma_1') = \sum_{a \in A(I)} \sigma_1'(I,a) \tilde{r}^t(I,a)
\end{equation}

\begin{lem} For any $I \in \mathcal{I}_1$, $\sigma_1' \in \Sigma_1$, and $t \in \{1,...,T\}$
\begin{equation}
\sum_{I' \in D(I)} \pi_1^{\sigma_1'}(I,I') \tilde{r}^t(I',\sigma_1') = \tilde{v}^t(I,(\sigma_1',\sigma_2^t)) - \tilde{v}^t(I,\sigma^t)
\end{equation}
\label{lemregsum}
\end{lem}

\begin{proof}
First note that
\begin{align*}
\tilde{r}^t(I,\sigma_1') &= \sum_{a \in A(I)} \sigma_1'(I,a) \left[ \tilde{v}^t(I,\sigma^t,a) - \tilde{v}^t(I,\sigma^t) \right]\\
&= \sum_{a \in A(I)} \sigma_1'(I,a) \tilde{v}^t(I,\sigma^t,a) - \sum_{a \in A(I)}\sigma_1^t(I,a)\tilde{v}^t(I,\sigma^t,a)\\
&= \sum_{a \in A(I)} \left(\sigma_1'(I,a) - \sigma_1^t(I,a)\right) \tilde{v}^t(I,\sigma^t,a)
\end{align*}
We now proceed by strong induction on $|D(I)|$. For the base case, consider $I\in\mathcal{I}_1$ where $D(I) = \{I\}$. Using that $|Succ(I)|=0$:
\begin{align*}
\sum_{I' \in D(I)} \pi_1^{\sigma_1'}(I,I') \tilde{r}^t(I',\sigma_1') &= \tilde{r}^t(I',\sigma_1')\\
&= \sum_{a \in A(I)}(\sigma_1'(I,a) - \sigma_1^t(I,a))\tilde{v}^t(I',\sigma^t,a)\\
&= \sum_{a \in A(I)}(\sigma_1'(I,a) - \sigma_1^t(I,a)) \left[\sum_{z \in Z[Ia]}\pi_{-1}^{\sigma_2^t}(z)u(z) - c^t(I,a) \right]\\
&= v^t(I,(\sigma_1',\sigma_2^t)) - \sum_{a\in A(I)} \sigma_1'(I,a) c^t(I,a) - \left[v^t(I,\sigma^t) - \sum_{a\in A(I)} \sigma_1^t(I,a) c^t(I,a) \right]\\
&= \tilde{v}^t(I,(\sigma_1',\sigma_2^t)) - \tilde{v}^{t}(I,\sigma^t)
\end{align*}
For the inductive step, we consider some $I \in \mathcal{I}_1$ and assume that the lemma holds for all $I'$ such that $|D(I')| < |D(I)|$. In particular, we assume it holds for each $I' \in Succ(I)$.
\begin{align*}
\sum_{I' \in D(I)} \pi_1^{\sigma_1'}(I,I') \tilde{r}^t(I',\sigma_1') &= \tilde{r}^t(I,\sigma_1') + \sum_{I' \in Succ(I)} \pi^{\sigma_1'}(I,I') \sum_{I'' \in D(I')} \pi_1^{\sigma_1'}(I',I'') \tilde{r}^t(I'',\sigma_1')\\
&= \tilde{r}^t(I,\sigma_1') + \sum_{I' \in Succ(I)} \pi^{\sigma_1'}(I,I') \left( \tilde{v}^{\sigma_1',\sigma_2^t}_{\lambda^t}(I') - \tilde{v}^t(I') \right)\\
&= \tilde{r}^t(I,\sigma_1') + \sum_{a \in A(I)} \sigma_1'(I,a) \left[ \sum_{I' \in Succ(I,a)} \tilde{v}^t(I',(\sigma_1',\sigma_2^t)) - \sum_{I' \in Succ(I,a)} \tilde{v}^t(I',\sigma^t) \right]\\
&= \tilde{r}^t(I,\sigma_1') + \sum_{a \in A(I)} \sigma_1'(I,a) \left[ \left( \tilde{v}^t(I,(\sigma_1',\sigma_2^t),a) - \sum_{z \in Z^1[Ia]} \pi_{-1}^{\sigma_2^t}(z)u(z) + c^t(I,a) \right) \right.\\
&\hspace{150pt}\left.- \left( \tilde{v}^t(I,\sigma^t,a) - \sum_{z \in Z^1[Ia]} \pi_{-1}^{\sigma_2^t}(z)u(z) +  c^t(I,a)\right)\right]\\
&= \tilde{r}^t(I,\sigma_1') + \tilde{v}^t(I,(\sigma_1',\sigma_2^t)) - \sum_{a \in A(I)} \sigma_1'(I,a)\tilde{v}^t(I,\sigma^t,a)\\
&= \sum_{a \in A(I)}\sigma_1'(I,a) \tilde{v}^t(I,\sigma^t,a) - \tilde{v}^t(I,\sigma^t) + \tilde{v}^t(I,(\sigma_1',\sigma_2^t)) - \sum_{a \in A(I)} \sigma_1'(I,a) \tilde{v}^t(I,\sigma^t,a)\\
&= \tilde{v}^t(I,(\sigma_1',\sigma_2^t)) - \tilde{v}^t(I,\sigma^t)
\end{align*}
\end{proof}

\begin{lem}
For any $T \in \mathbb{N}$, for any sequences of strategies $\sigma_1^1,...,\sigma_1^T$ and $\sigma_2^1,...,\sigma_2^T$, and any sequence of vectors $\lambda^1,...,\lambda^T$ each $\in \mathbb{R}^k$,
\begin{equation}
\max_{\sigma_1^* \in \Sigma_1} \sum_{t=1}^T \left[ \tilde{u}(\sigma_1^*,\sigma_2^t,\lambda^t) - \tilde{u}(\sigma_1^t,\sigma_2^t,\lambda^t) \right] \leq \max_{\sigma_1^* \in \Sigma_1} \sum_{I \in \mathcal{I}_1} \pi_1^{\sigma_1^*}(I) \max_{a \in A(I)} \sum_{t=1}^T \tilde{r}^t(I,a)
\end{equation}
\label{lemimmreg}
\end{lem}

\begin{proof}
We can modify the original game to add a dummy history $\emptyset$ before any other actions where $P(\emptyset) = 1$, $|A(\emptyset)| = 1$, and $\emptyset$ forms a singleton information set. This doesn't change the values of the game in any way. Then for any strategy $\sigma_1$ and any $t \in \{1,...,T\}$, we have that
\begin{equation}
\tilde{v}^t(\emptyset,\sigma) = u(\sigma_1,\sigma_2^t) - \sum_{I \in \mathcal{I}_1} \sum_{a \in A(I)} \pi_1^{\sigma_1}(Ia) c^t(I,a)
\end{equation}

Then, using the convexity of each $f_i$, we can show for any $\sigma_1' \in \Sigma_1$ that:
\begin{align*}
\tilde{u}(\sigma_1',\sigma_2^t,\lambda^t) - \tilde{u}(\sigma_1^t,\sigma_2^t,\lambda^t) &= u(\sigma_1',\sigma_2^t) - \sum_{i=1}^k \lambda^t_i f_i(\Psi(\sigma_1')) - u(\sigma_1^t,\sigma_2^t) + \sum_{i=1}^k \lambda^t_i f_i(\Psi(\sigma_1^t))\\
&= u(\sigma_1',\sigma_2^t) - u(\sigma_1^t,\sigma_2^t) + \sum_{i=1}^k \lambda^t_i \left( f_i(\Psi(\sigma_1^t)) - f_i(\Psi(\sigma_1'))\right)\\
&\leq u(\sigma_1',\sigma_2^t) - u(\sigma_1^t,\sigma_2^t) + \sum_{i=1}^k \lambda^t_i \left( \langle \nabla f_i(\Psi(\sigma_1^t)), \Psi(\sigma_1^t) - \Psi(\sigma_1') \rangle \right)\\
&= u(\sigma_1',\sigma_2^t) - \sum_{I \in \mathcal{I}_1}\sum_{a \in A(I)} \pi_1^{\sigma_1'}(Ia) c^t(I,a) - \left(u(\sigma_1^t,\sigma_2^t) - \sum_{I \in \mathcal{I}_1}\sum_{a \in A(I)} \pi_1^{\sigma_1^t}(Ia) c^t(I,a)\right)\\
&= \tilde{v}^t(\emptyset,(\sigma_1',\sigma_2^t)) - \tilde{v}^t(\emptyset,(\sigma_1^t,\sigma_2^t))
\end{align*}

We then apply Lemma~\ref{lemregsum} to give us
\begin{align*}
\sum_{t=1}^T\left[\tilde{u}(\sigma_1',\sigma_2^t,\lambda^t) - \tilde{u}(\sigma_1^t,\sigma_2^t,\lambda^t)\right] &= \sum_{t=1}^T\left[\tilde{v}^t(\emptyset,(\sigma_1',\sigma_2^t)) - \tilde{v}^t(\emptyset,(\sigma_1^t,\sigma_2^t))\right]\\
&= \sum_{t=1}^T \sum_{I \in \mathcal{I}_1} \pi_1^{\sigma_1'}(I) \tilde{r}^t(I,\sigma_1')\\
&= \sum_{t=1}^T \sum_{I \in \mathcal{I}_1} \pi_1^{\sigma_1'}(I) \sum_{a \in A(I)} \sigma_1'(I,a) \tilde{r}^t(I,a)\\
&= \sum_{I \in \mathcal{I}_1} \pi_1^{\sigma_1'}(I) \sum_{a \in A(I)} \sigma_1'(I,a) \sum_{t=1}^T \tilde{r}^t(I,a)\\
&\leq \sum_{I \in \mathcal{I}_1} \pi_1^{\sigma_1'}(I) \max_{a \in A(I)} \sum_{t=1}^T \tilde{r}^t(I,a)
\end{align*}
The last statement follows because $\sum_{a \in A(I)} \sigma_1'(I,a) = 1$.

This holds for all $\sigma_1' \in \Sigma_1$, so the lemma must hold. 
\end{proof}

\begin{thm}
For any $T \in \mathbb{N}$, any sequence of player 2 strategies $\sigma_2^1,...,\sigma_2^T$, and any sequence of vectors $\lambda^1,...,\lambda^T$ each $\in [0,\beta]^k$, if CCFR is used to select the sequence of strategies $\sigma_1^1,...,\sigma_1^T$, then
\begin{equation}
\max_{\sigma_1^* \in \Sigma_1} \frac{1}{T} \sum_{t=1}^T \left[ \tilde{u}(\sigma_1^*,\sigma_2^t,\lambda^t) - \tilde{u}(\sigma_1^t,\sigma_2^t,\lambda^t) \right] \leq \frac{\left(\Delta_u + 2k\beta F\right)M \sqrt{|A_1|}}{\sqrt{T}}
\end{equation}
where $\Delta_u = \max_z u(z) - \min_z u(z)$ is the range of possible utilities, $k$ is the number of constraints, $F=\max_{\vect{x},i}||\nabla f_i(x)||_1$ is a bound on the subgradients, and $M$ is a game-specific constant.
\label{thmreg}
\end{thm}

\begin{proof}
For any information set $I \in \mathcal{I}_1$, $\sigma_1^t(I,a)$ is selected for each $a \in A(I)$ using regret matching with action values $\tilde{v}^t(I,a)$. The bound for regret matching then guarantees
\begin{equation*}
\max_{a \in A(I)} \frac{1}{T}\sum_{t=1}^T \tilde{r}^t(I,a) \leq \frac{\sqrt{|A(I)|}}{\sqrt{T}} \max_{\substack{t \in \{1,...,T\}\\a,a' \in A(I)}} (\tilde{v}^t(I,\sigma^t,a) - \tilde{v}^t(I,\sigma^t,a'))
\end{equation*}
Using Lemma~\ref{lemcfv}, we bound the max term as
\begin{align*}
\max_{\substack{t \in \{1,...,T\}\\a,a' \in A(I)}} &(\tilde{v}^t(I,\sigma^t,a) - \tilde{v}^t(I,\sigma^t,a'))\\
 &= \max_{\substack{t \in \{1,...,T\}\\a,a' \in A(I)}} \left( v(I,\sigma^t,a) - v(I,\sigma^t,a') - c^t(I,a) - \sum_{\substack{I' \in D(I)\\Ia \sqsubseteq I'}} \sum_{a' \in A(I')} \pi_1^{\sigma_1^t}(Ia,I'a') c^t(I',a')\right.\\
 &\qquad\left. + c^t(I,a') - \sum_{\substack{I' \in D(I)\\Ia' \sqsubseteq I'}} \sum_{a'' \in A(I')} \pi_1^{\sigma_1^t}(Ia',I'a'') c^t(I',a'')\right)\\
 &\leq \max_{\substack{\sigma\\a,a' \in A(I)}} \left( v(I,\sigma,a) - v(I,\sigma,a')\right)\\
 &\qquad + \max_{\substack{\sigma\\a,a' \in A(I)}} \left(|c^t(I,a)| + \sum_{\substack{I' \in D(I)\\Ia \sqsubseteq I'}} \sum_{a' \in A(I')}|c^t(I,a')| + |c^t(I,a')| + \sum_{\substack{I' \in D(I)\\Ia' \sqsubseteq I'}} \sum_{a'' \in A(I')}|c^t(I,a'')|\right)\\
 &\leq \max_{z,z' \in Z} \left(u(z) - u(z')\right) + 2\sum_{I \in \mathcal{I}_1}\sum_{a \in A(I)} |c^t(I,a)|\\
 &\leq \Delta_u + 2\sum_{i=1}^k \lambda^t_i ||\nabla f(\Psi(\sigma_1^t)) ||_1\\
 &\leq \Delta_u + 2k\beta F
\end{align*}
We can now use this bound with Lemma~\ref{lemimmreg} to give
\begin{align*}
\max_{\sigma_1^* \in \Sigma_1} \frac{1}{T} \sum_{t=1}^T \left[ \tilde{u}(\sigma_1^*,\sigma_2^t,\lambda^t) - \tilde{u}(\sigma_1^t,\sigma_2^t,\lambda^t) \right] &\leq \frac{\sqrt{|A(I)|}}{\sqrt{T}} \left(\Delta_u + 2k\beta F\right)\max_{\sigma_1^* \in \Sigma_1} \sum_{I \in \mathcal{I}_1} \pi_1^{\sigma_1^*}(I)\\
&\leq \frac{\left(\Delta_u + 2k\beta F\right)M \sqrt{|A_1|}}{\sqrt{T}}
\end{align*}
\end{proof}

Define $R^T_{\lambda}(\beta)$ to the regret of player 2 for choosing the sequence $\lambda^1,...,\lambda^T$ against the set of possible $\lambda^* \in [0,\beta]^k$:
\begin{align*}
R^T_{\lambda}(\beta) &= \max_{\lambda^* \in [0,\beta]^k} \sum_{t=1}^T \left[ - \tilde{u}(\sigma_1^t,\sigma_2^t,\lambda^*) + \tilde{u}(\sigma_1^t,\sigma_2^t,\lambda^t) \right]
\end{align*}
$R^T_{\lambda}(\beta)$ can be minimized by using, e.g., projected gradient descent to choose $\lambda^t$ for each $t$.

Define $f^*=\min_\vect{x} \sum_{i=1}^k (f_i(\vect{x}))^+$ to be the minimum sum of constraint violations (if the constraints are feasible, $f^*=0$), and define $\mathcal{X}_{f^*} = \{ \vect{x} \in \mathcal{X} \mid \sum_{i=1}^k (f_i(\vect{x}))^+ = f^* \}$ to be the set of strategies that achieve $f^*$. Note that these are well-defined and $\mathcal{X}_{f^*}$ is non-empty, since each $f_i$ is continuous and $\mathcal{X}$ is compact. Define $\Sigma_1^{f^*} = \{ \sigma_1 \in \Sigma_1 \mid \Psi(\sigma_1) \in \mathcal{X}_{f^*} \}$. $\Sigma_1^{f^*}$ is non-empty as $\Psi$ is surjective.

\begin{thm}
For any $T \in \mathbb{N}$, any series of vectors $\lambda^1,...,\lambda^T$ each $\in [0,\beta]^k$, and any $f_1,...,f_k$ which are convex, continuous constraints, if CCFR is used to select the sequence of strategies $\sigma_1^1,...,\sigma_1^T$ and CFR is used to select the sequence of strategies $\sigma_2^1,...,\sigma_2^T$, then the following hold:
\begin{align}
\sum_{i=1}^k (f_i(\Psi(\overline{\sigma}_1^t)))^+ - f^* &\leq \frac{R_\lambda^T(\beta)}{\beta} + \frac{\left(\Delta_u + 2k\beta F\right)M \sqrt{|A_1|}}{\beta\sqrt{T}} + \frac{\Delta_u}{\beta}\label{eqcon}\\
\max_{\sigma_1^* \in \Sigma_1^{f^*}} u(\sigma_1^*,\overline{\sigma}_2^T) - \min_{\sigma_2^* \in \Sigma_2} u(\overline{\sigma}_1^T,\sigma_2^*) &\leq \frac{4\left(\Delta_u + k\beta F\right)M \sqrt{|A_1|}}{\sqrt{T}} + 2R_\lambda^T(\beta)\label{eqexp}
\end{align}
\label{thm:infcon}
\end{thm}

\begin{proof}
First, we show (\ref{eqcon}).

Let $\sigma_1^f$ be any strategy in $\Sigma_1^{f^*}$. Then for any $t$
\begin{align*}
\tilde{u}(\sigma_1^f,\sigma_2^t,\lambda^t) - \tilde{u}(\sigma_1^t,\sigma_2^t,\lambda^t) &= u(\sigma_1^f,\sigma_2^t) - \sum_{i=1}^k \lambda_i^t f(\Psi(\sigma_1^f)) - u(\sigma_1^t,\sigma_2^t) + \sum_{i=1}^k \lambda_i^t f(\Psi(\sigma_1^t))\\
&\geq \sum_{i=1}^k \lambda_i^t f(\Psi(\sigma_1^t)) - \sum_{i=1}^k \beta (f(\Psi(\sigma_1^f)))^+ - (u(\sigma_1^t,\sigma_2^t) - u(\sigma_1^f,\sigma_2^t))\\
&\geq  \sum_{i=1}^k \lambda_i^t f(\Psi(\sigma_1^t)) - \beta f^* - \Delta_u
\end{align*}
Applying Theorem~\ref{thmreg} then gives us
\begin{equation}
\frac{1}{T}\sum_{t=1}^T \sum_{i=1}^k \lambda_i^t f(\Psi(\sigma_1^t)) \leq \frac{\left(\Delta_u + 2k\beta F\right)M \sqrt{|A_1|}}{\sqrt{T}} + \beta f^* + \Delta_u\label{eq}
\end{equation}
Consider now $R^T_\lambda(\beta)$. We have
\begin{align*}
R^T_{\lambda}(\beta) &= \max_{\lambda^* \in [0,\beta]^k} \frac{1}{T} \sum_{t=1}^T \left[ - \tilde{u}(\sigma_1^t,\sigma_2^t,\lambda^*) + \tilde{u}(\sigma_1^t,\sigma_2^t,\lambda^t) \right]\\
&= \max_{\lambda^* \in [0,\beta]^k} \frac{1}{T} \sum_{t=1}^T \sum_{i=1}^k \lambda_i^* f(\Psi(\sigma_1^t)) - \frac{1}{T} \sum_{t=1}^T \sum_{i=1}^k \lambda_i^t f(\Psi(\sigma_1^t))\\
&\geq \max_{\lambda^* \in [0,\beta]} \sum_{i=1}^k \lambda_i^* f(\Psi(\overline{\sigma}_1^T)) - \frac{1}{T} \sum_{t=1}^T\sum_{i=1}^k \lambda_i^t f(\Psi(\sigma_1^t))
\end{align*}
where the last step uses Jensen's inequality. Combining this with (\ref{eq}) gives us
\begin{equation}
\max_{\lambda^* \in [0,\beta]} \sum_{i=1}^k \lambda_i^* f(\Psi(\overline{\sigma}_1^T)) \leq R_\lambda^T(\beta) + \frac{\left(\Delta_u + 2k\beta F\right)M \sqrt{|A_1|}}{\sqrt{T}} + \beta f^* + \Delta_u
\end{equation}
Define $\lambda^{\overline{\sigma}_1^T}$ such that, for each index $j \in \{1,...,k\}$, we have $\lambda^{\overline{\sigma}_1^T}_j = \beta$ if $f_j(\overline{\sigma}_1^T) > 0$ and $\lambda^{\overline{\sigma}_1^T}_j = 0$ otherwise. Then we have
\begin{align*}
\beta \sum_{i=1}^k (f_i(\Psi(\overline{\sigma}_1^T)))^+ &= \sum_{i=1}^k \lambda^{\overline{\sigma}^T_1}_i f_i(\Psi(\overline{\sigma}_1^T))\\
&\leq \max_{\lambda^* \in [0,\beta]^k} \sum_{i=1}^k \lambda_i^* f(\Psi(\overline{\sigma}_1^T))\\
&\leq \frac{\left(\Delta_u + 2k\beta F\right)M \sqrt{|A_1|}}{\sqrt{T}} + \beta f^* + \Delta_u + R_\lambda^T(\beta)\\
\sum_{i=1}^k (f_i(\Psi(\overline{\sigma}^T_1)))^+ - f^* &\leq \frac{R_\lambda^T(\beta)}{\beta} + \frac{\left(\Delta_u + 2k\beta F\right)M \sqrt{|A_1|}}{\beta\sqrt{T}} + \frac{\Delta_u}{\beta} 
\end{align*}

Next, we show (\ref{eqexp}). We break this into parts, showing that each player's average strategy does almost as well as a best response against the other player's actual strategy.

Consider first how well player 1 could do be deviating to any other strategy $\sigma_1' \in \Sigma_1^{f^*}$:
\begin{align*}
u(\sigma_1',\overline{\sigma}_2^T) &= \frac{1}{T} \sum_{t=1}^T u(\sigma_1',\sigma_2^t)\\
&= \frac{1}{T} \sum_{t=1}^T \tilde{u}(\sigma_1',\sigma_2^t,\lambda^t) + \sum_{i=1}^k \lambda^t_i f_i(\Psi(\sigma_1'))\\
&\leq \frac{1}{T} \max_{\sigma_1^* \in \Sigma_1} \sum_{t=1}^T \tilde{u}(\sigma_1',\sigma_2^t,\lambda^t) + \beta f^*\\
&= \frac{1}{T} \sum_{t=1}^T \tilde{u}(\sigma_1^t,\sigma_2^t,\lambda^t) + \tilde{R}_1^T + \beta f^*\\
&= \frac{1}{T} \sum_{t=1}^T u(\sigma_1^t,\sigma_2^t) - \frac{1}{T} \sum_{t=1}^T \sum_{i=1}^k \lambda_i^t f_i(\Psi(\sigma_1^t)) + \tilde{R}_1^T + \beta f^*\\
&= \frac{1}{T} \min_{\sigma_2^*} \sum_{t=1}^T u(\sigma_1^t,\sigma_2^*) + R_2^T + R_\lambda^T(\beta) - \frac{1}{T} \max_{\lambda^*}\sum_{t=1}^T \sum_{i=1}^k \lambda_i^* f_i(\Psi(\sigma_1^t)) + \tilde{R}_1^T + \beta f^*\\
&= \frac{1}{T} \min_{\sigma_2^*} \sum_{t=1}^T u(\sigma_1^t,\sigma_2^*) - \max_{\lambda^*}\sum_{i=1}^k \lambda_i^* \left(\frac{1}{T}\sum_{t=1}^T f_i(\Psi(\sigma_1^t))\right) + \tilde{R}_1^T + R_2^T + R_\lambda^T(\beta)\\
&\leq u(\overline{\sigma}_1^T,\overline{\sigma}_2^T) - \sum_{i=1}^k \beta \left(\frac{1}{T}\sum_{t=1}^T f_i(\Psi(\sigma_1^t))\right)^+ + \tilde{R}_1^T + R_2^T + R_\lambda^T(\beta) + \beta f^*
\intertext{Using Jensen's inequality:}
&\leq u(\overline{\sigma}_1^T,\overline{\sigma}_2^T) - \beta \sum_{i=1}^k  \left(f_i(\Psi(\overline{\sigma}_1^T))\right)^+ + \tilde{R}_1^T + R_2^T + R_\lambda^T(\beta) + \beta f^*\\
&\leq u(\overline{\sigma}_1^T,\overline{\sigma}_2^T) - \beta \min_{\vect{x} \in \mathcal{X}} \sum_{i=1}^k  \left(f_i(\vect{x})\right)^+ + \tilde{R}_1^T + R_2^T + R_\lambda^T(\beta) + \beta f^*\\
&= u(\overline{\sigma}_1^T,\overline{\sigma}_2^T) + \tilde{R}_1^T + R_2^T + R_\lambda^T(\beta)
\end{align*}
This is true for each $\sigma' \in \Sigma_1^{f^*}$, so we get
\begin{equation}
\max_{\sigma_1^* \in \Sigma_1^{f^*}} u(\sigma_1^*,\overline{\sigma}_2^T) - u(\overline{\sigma}_1^T,\overline{\sigma}_2^T) \leq \tilde{R}_1^T + R_2^T + R_\lambda^T(\beta) \label{eqp1}
\end{equation}

Next we look at how well player 2 could do by deviating to any other strategy $\sigma_2'$.
\begin{align*}
{-u}(\overline{\sigma}_1^T,\sigma_2') &= {-\frac{1}{T}}\sum_{t=1}^T u(\sigma_1^t,\sigma_2')\\
&\leq {-\frac{1}{T}} \min_{\sigma_2^*} \sum_{t=1}^T u(\sigma_1^t,\sigma_2^*)\\
&= {-\frac{1}{T}} \sum_{t=1}^T u(\sigma_1^t,\sigma_2^t) + R_2^T\\
&= {-\frac{1}{T}} \sum_{t=1}^T \tilde{u}(\sigma_1^t,\sigma_2^t,\lambda^t) - \frac{1}{T}\sum_{t=1}^T \sum_{i=1}^k \lambda_i^t f_i(\Psi(\sigma_1^t)) + R_2^T\\
&\leq \tilde{R}_1^T - \frac{1}{T} \max_{\sigma_1^*}\sum_{t=1}^T \tilde{u}(\sigma_1^*,\sigma_2^t,\lambda^t) - \frac{1}{T}\sum_{t=1}^T \sum_{i=1}^k \lambda_i^t f_i(\Psi(\sigma_1^t)) + R_2^T\\
&\leq  {-\frac{1}{T}} \sum_{t=1}^T \tilde{u}(\overline{\sigma}_1^T,\sigma_2^t,\lambda^t) + R_\lambda^T(\beta) - \max_{\lambda^* \in [0,\beta]^k}\frac{1}{T}\sum_{t=1}^T \sum_{i=1}^k \lambda_i^* f_i(\Psi(\sigma_1^t)) + \tilde{R}_1^T + R_2^T\\
&= {-u}(\overline{\sigma}_1^T,\overline{\sigma}_2^T) + \frac{1}{T}\sum_{t=1}^T \sum_{i=1}^k \lambda_i^t f_i(\Psi(\overline{\sigma}_1^T)) - \max_{\lambda^* \in [0,\beta]^k}\frac{1}{T}\sum_{t=1}^T \sum_{i=1}^k \lambda_i^* f_i(\Psi(\sigma_1^t)) + \tilde{R}_1^T + R_2^T + R_\lambda^T\\
&\leq {-u}(\overline{\sigma}_1^T,\overline{\sigma}_2^T) + \frac{1}{T}\sum_{t=1}^T \sum_{i=1}^k \lambda_i^t f_i(\Psi(\overline{\sigma}_1^T)) - \frac{1}{T}\sum_{t=1}^T \sum_{i=1}^k \overline{\lambda}_i^T f_i(\Psi(\sigma_1^t)) + \tilde{R}_1^T + R_2^T + R_\lambda^T\\
&= {-u}(\overline{\sigma}_1^T,\overline{\sigma}_2^T) + \sum_{i=1}^k \overline{\lambda}_i^T \left(f_i(\Psi(\overline{\sigma}_1^T)) - \frac{1}{T} \sum_{t=1}^T f_i(\Psi(\sigma_1^t))\right) + \tilde{R}_1^T + R_2^T + R_\lambda^T\\
&\leq {-u}(\overline{\sigma}_1^T,\overline{\sigma}_2^T) + \tilde{R}_1^T + R_2^T + R_\lambda^T
\end{align*}
In the last step we use Jensen's inequality, along with the fact that $\overline{\lambda}_i^T \geq 0$ for all $i$.

Again, this holds for all $\sigma_2'$, so necessarily
\begin{equation}
{-\min_{\sigma_2^* \in \Sigma_2}} u(\overline{\sigma}_1^T,\sigma_2^*) + u(\overline{\sigma}_1^T,\overline{\sigma}_2^T) \leq \tilde{R}_1^T + R_2^T + R_\lambda^T(\beta)\label{eqp2}
\end{equation}

Adding together (\ref{eqp1}) and (\ref{eqp2}), then substituting bounds on $\tilde{R}_1^T$ (from Theorem~\ref{eqp1}) and $R_2^T$ (from CFR bound) gives the result.

\end{proof}

Theorem~1 follows from the second part of Theorem~\ref{thm:infcon} when the constraints are simultaneously feasible, in which case $\sigma_1^{f^*}$ becomes the feasible set. Theorem~2 follows from the first part of Theorem~\ref{thm:infcon} when the constraints are simultaneously feasible, noting that $f_i(\Psi(\overline{\sigma}_1^t)) \leq \sum_{i=1}^k (f_i(\Psi(\overline{\sigma}_1^t)))^+$ for all $i$.

We now prove Theorem~3, which we restate here for clarity.
\addtocounter{thm}{-3}
\begin{thm}
Assume that $f_1,...,f_k$ satisfy a constraint qualification such as Slater's condition, and define $\vect{\lambda}^*$ to be a finite solution for $\vect{\lambda}$ in the Lagrangian optimization. Then if $\beta$ is chosen such that $\beta > \lambda_i^*$ for all $i$, and CCFR and CFR are used to respectively select the strategy sequences $\sigma_1^1,...,\sigma_1^T$ and $\sigma_2^1,...,\sigma_2^T$, the following holds:
\begin{align}
&f_i(\Psi(\overline{\sigma}_1^T)) \leq \frac{R_\lambda^T(\beta)}{\beta-\lambda^*_i} + \frac{2\left(\Delta_u + k\beta F\right)M \sqrt{|A|}}{(\beta-\lambda^*_i)\sqrt{T}}\hspace{5em}\forall i \in \{1,...,k\}
\end{align}
\end{thm}
\addtocounter{thm}{2}

\begin{proof}
We establish the theorem by proving upper and lower bounds on the gap in exploitability between $\overline{\sigma}_1^T$ and the optimal strategy in the modified game with utility function $\tilde{u}$ and $\vect{\lambda} \in [0,\beta]$, and comparing the bounds.

First, we show that the exploitability gap for any strategy can be lower bounded in terms of the strategy's constraint violations. Because we assume a constraint qualification, there is a finite $\vect{\lambda}^*$ that is a solution the optimization
\begin{equation}
\vect{\lambda}^* \in \argmin_{\vect{\lambda} \geq 0} \max_{\sigma_1 \in \Sigma_1} \min_{\sigma_2 \in \Sigma_2} \tilde{u}(\sigma_1,\sigma_2,\vect{\lambda}).
\end{equation}
Note that by assuming a constraint qualification, strong duality holds and we can rearrange the order of min and max terms in the optimization without changing the solution. Consider an arbitrary $\sigma_1' \in \Sigma_1$:
\begin{align*}
\min_{\vect{\lambda} \in [0,\beta]^k} \min_{\sigma_2 \in \Sigma_2} u(\sigma_1',\sigma_2) - \sum_{i=1}^k\lambda_i f_i(\Psi(\sigma_1')) &= \min_{\sigma_2 \in \Sigma_2} u(\sigma_1',\sigma_2) - \sum_{i=1}^k\beta(f_i(\Psi(\sigma_1')))^+\\
 &= \min_{\sigma_2 \in \Sigma_2} u(\sigma_1',\sigma_2) - \sum_{i=1}^k\lambda^*_i (f_i(\Psi(\sigma_1')))^+ - \sum_{i=1}^k(\beta-\lambda^*_i)(f_i(\Psi(\sigma_1')))^+\\
&\leq \max_{\sigma_1 \in \Sigma_1} \min_{\sigma_2 \in \Sigma_2} \left[u(\sigma_1,\sigma_2) - \sum_{i=1}^k \lambda^*_i f_i(\Psi(\sigma_1))\right] - \sum_{i=1}^k(\beta-\lambda^*_i)(f_i(\Psi(\sigma_1')))^+\\
&= \min_{\vect{\lambda} \geq 0} \max_{\sigma_1 \in \Sigma_1} \min_{\sigma_2 \in \Sigma_2} \left[u(\sigma_1,\sigma_2) - \sum_{i=1}^k \lambda_i f_i(\Psi(\sigma_1))\right] - \sum_{i=1}^k(\beta-\lambda^*_i)(f_i(\Psi(\sigma_1')))^+\\
&\leq \min_{\vect{\lambda} \in [0,\beta]^k} \max_{\sigma_1 \in \Sigma_1} \min_{\sigma_2 \in \Sigma_2} \left[u(\sigma_1,\sigma_2) - \sum_{i=1}^k \lambda_i f_i(\Psi(\sigma_1))\right] - \sum_{i=1}^k(\beta-\lambda^*_i)(f_i(\Psi(\sigma_1')))^+\\
&= \max_{\sigma_1 \in \Sigma_1} \min_{\vect{\lambda} \in [0,\beta]^k}  \min_{\sigma_2 \in \Sigma_2} \left[u(\sigma_1,\sigma_2) - \sum_{i=1}^k \lambda_i f_i(\Psi(\sigma_1))\right] - \sum_{i=1}^k(\beta-\lambda^*_i)(f_i(\Psi(\sigma_1')))^+
\end{align*}
In the last line we use strong duality, which follows from the assumed constraint qualification. This holds for any $\sigma_1'$, and thus we have established that
\begin{equation}
\max_{\sigma_1 \in \Sigma_1}  \min_{\vect{\lambda} \in [0,\beta]^k} \min_{\sigma_2 \in \Sigma_2} \tilde{u}(\sigma_1,\sigma_2,\vect{\lambda}) - \min_{\vect{\lambda} \in [0,\beta]^k} \min_{\sigma_2 \in \Sigma_2} \tilde{u}(\overline{\sigma}^T_1,\sigma_2,\vect{\lambda}) \geq \sum_{i=1}^k (\beta - \lambda^*_i) (f_i(\Psi(\overline{\sigma}_1^T)))^+
\label{eq:gaplb}
\end{equation}

We now show an upper bound on exploitability based on the regrets of the sequences used in CCFR:
\begin{align*}
\min_{\vect{\lambda} \in [0,\beta]^k} \min_{\sigma_2 \in \Sigma_2} \left[ u(\overline{\sigma}_1^T,\sigma_2) - \sum_{i=1}^k\lambda_i f_i(\Psi(\overline{\sigma}_1^T))\right] &\geq \min_{\vect{\lambda} \in [0,\beta]^k} \min_{\sigma_2 \in \Sigma_2} \sum_{t=1}^k \left[u(\sigma_1^t,\sigma_2) - \sum_{i=1}^k\lambda_i f_i(\Psi(\sigma_1^t))\right]\\ &\qquad\qquad\qquad\text{by Jensen's inequality}\\
&\geq \sum_{t=1}^T \left[ u(\sigma_1^t,\sigma_2^t) - \sum_{i=1}^k\lambda^t_i f_i(\Psi(\sigma_1^t))\right] - R^T_2 - R^T_\lambda(\beta)\\
&\geq \max_{\sigma_1 \in \Sigma_1} \sum_{t=1}^T \left[ u(\sigma_1,\sigma_2^t) - \sum_{i=1}^k\lambda^t_i f_i(\Psi(\sigma_1))\right] - \tilde{R}^T_1 - R^T_2 - R^T_\lambda(\beta)\\
&\geq \max_{\sigma_1 \in \Sigma_1} \min_{\vect{\lambda} \in [0,\beta]^k}  \min_{\sigma_2 \in \Sigma_2} \left[u(\sigma_1,\sigma_2) - \sum_{i=1}^k \lambda_i f_i(\Psi(\sigma_1))\right] - \tilde{R}^T_1 - R^T_2 - R^T_\lambda(\beta)
\end{align*}
This gives us an upper bound of:
\begin{equation}
\max_{\sigma_1 \in \Sigma_1}  \min_{\vect{\lambda} \in [0,\beta]^k} \min_{\sigma_2 \in \Sigma_2} \tilde{u}(\sigma_1,\sigma_2,\vect{\lambda}) - \min_{\vect{\lambda} \in [0,\beta]^k} \min_{\sigma_2 \in \Sigma_2} \tilde{u}(\overline{\sigma}^T_1,\sigma_2,\vect{\lambda}) \leq \tilde{R}^T_1 + R^T_2 + R^T_\lambda(\beta)
\label{eq:gapub}
\end{equation}

Comparing to (\ref{eq:gapub}), this gives us that, for any $j\in\{1,...,k\}$
\begin{align}
\sum_{i=1}^k (\beta - \lambda^*_i) (f_i(\Psi(\overline{\sigma}_1^T)))^+ &\leq \tilde{R}^T_1 + R^T_2 + R^T_\lambda(\beta)\\
(\beta - \lambda^*_j) (f_j(\Psi(\overline{\sigma}_1^T)))^+ &\leq \tilde{R}^T_1 + R^T_2 + R^T_\lambda(\beta)\\
f_j(\Psi(\overline{\sigma}_1^T)) &\leq \frac{\tilde{R}^T_1 + R^T_2 + R^T_\lambda(\beta)}{\beta - \lambda^*_j}
\end{align}
In the last line we use the assumption that $\beta > \lambda^*_j$. Substituting in the CCFR and CFR bounds for $\tilde{R}^T_1$ and $R^T_2$ completes the proof.

\end{proof}

\section{Infeasible constraints}

By Theorem~\ref{thm:infcon}, CCFR still converges when the constraints are infeasible. Effectively, the constrained strategy set is relaxed so that the constrained player instead chooses from the set of strategies that minimize the total ($L_1$ norm) infeasibility of the constraints. Thus, (\ref{eqcon}) says that the constrained player's strategy isn't much more infeasible than the least infeasible strategy, and (\ref{eqexp}) says that the strategies are near-optimal given  this restriction.

The relaxation of the constrained set follows from the bounding of the Lagrange multipliers. In fact, with bounded multipliers, the Lagrangian optimization becomes equivalent to a regularized optimization.

\begin{lem}
For any optimization function $g\colon \mathcal{X} \to \mathbb{R}$, constraint set $f_1,...,f_k$, and bound $\beta > 0$ the following holds:
\begin{equation}
\min_{\vect{x} \in \mathcal{X}} \max_{0 \leq \vect{\lambda} \leq \beta} g(\vect{x}) + \sum_{i=1}^k \lambda_i f_i(\vect{x}) = \min_{\vect{x}\in \mathcal{X}} g(\vect{x}) + \beta\sum_{i=1}^k (f_i(\vect{x}))^+.
\end{equation}
\end{lem}

\begin{proof}
As $\beta > 0$, it follows that $0 \leq \vect{\lambda} \leq \beta$ is feasible and thus
\begin{align*}
\min_{\vect{x} \in \mathcal{X}} \max_{0 \leq \vect{\lambda} \leq \beta} g(\vect{x}) + \sum_{i=1}^k \lambda_i f_i(\vect{x}) &= \min_{\substack{\vect{x} \in \mathcal{X}\\\vect{\xi} \geq 0}} \max_{0 \leq \vect{\lambda}} g(\vect{x}) + \sum_{i=1}^k \lambda_i f_i(\vect{x}) + \sum_{i=1}^k (\beta - \lambda_i) \xi_i\\
&= \min_{\substack{\vect{x} \in \mathcal{X}\\\vect{\xi} \geq 0}} \max_{0 \leq \vect{\lambda}} g(\vect{x}) + \sum_{i=1}^k \lambda_i (f_i(\vect{x}) - \xi_i) + \beta\sum_{i=1}^k \xi_i\\
\intertext{For each $i$, necessarily $f_i(\vect{x}) \leq \xi_i$, or unbounded $\lambda_i$ could be chosen. Given this constraint, $\lambda_i = 0$.}
&= \min_{\substack{\vect{x} \in \mathcal{X}\\ \xi_i \geq (f_i(\vect{x}))^+}} g(\vect{x}) + \beta\sum_{i=1}^k \xi_i\\
&=  \min_{\vect{x}\in \mathcal{X}} g(\vect{x}) + \beta\sum_{i=1}^k (f_i(\vect{x}))^+ \qedhere
\end{align*}
\end{proof}

The fact the strategy set is relaxed to strategies that minimize infeasibility as measured by the $L_1$ norm in particuar is a consequence of bounding the $L_\infty$ norm of the Lagrange multipliers. If we cared more about maximum violation of any particular constraint, we could bound the $L_1$ norm of the multipliers instead.

\section{Constraints for opponent modeling}

In poker and similar games, the actions which constitute a history $h$ can be uniquely decomposed into two sets of \emph{private chance actions} $\our[h]$ and $\opp[h]$ which are only visible to ourselves or our opponents, respectively, and \emph{public actions} $\pub[h]$ which both players observe. In poker, the private chance actions are the private cards, consituting a \emph{hand}, dealt to each player---we will use this terminology. Let $\Priv_i$ and $\Priv_{-i}$ be the sets of possible hands (so $\Priv_i = \{ \our[h] \mid h \in H \}$), and let $S = \{ \pub[h] \mid h \in H \}$ be the set of possible public action sequences. For any $\our \in \Priv_i, \opp \in \Priv_{-i}, \pub \in \Pub$, let $h(\our,\opp,\pub)$ be the corresponding history if it exists and define $\pi(\our,\opp,\pub)=\pi(h(\our,\opp,\pub))$. If the history doesn't exist, we let $\pi_c(\our,\opp,\pub) = 0$ to simplify the analysis. For brevity, when $h(\our,\opp,\pub) \in Z$, we define $u(\our,\opp,\pub)=u(h(\our,\opp,\pub))$.

We assume that, for a given $h$, whether $h \in Z$ is uniquely determined by $\pub[h]$. We thus define $\Pubterm = \{ \pub \in \Pub \mid h(\cdot,\cdot,\pub) \in Z \}$ to be the set of public action sequences that lead to terminal histories. In the games we are interested in, there are terminal histories where the public actions fully determine the players' utilities, such that $u(\our,\opp,\pub) = u(\our',\opp',\pub)$ for all $\our,\our',\opp,\opp'$ (in poker, this occurs when a player folds, in which case the cards dealt do not matter). We define $\Pubtermu \subseteq \Pubterm$ to be the set of public action sequences that lead to such terminal histories, and define $u(\pub[z]) =u(z)$ for $z \in \Pubtermu$.

We introduce a method for opponent modeling from \emph{partially observed games}. In this setting, the game is played until a terminal history $z$ is reached, then $\our[z]$ and $\opp[z]$ are publicly revealed only if they affect the game outcome (i.e. $\pub[z] \notin \Pubtermu$). Thus, while we always observe $\our[z]$ and $\pub[z]$, we only sometimes observe $\opp[z]$. We wish to use these observed games to model the opponent, but without observing $\opp[z]$ we do not have direct samples of their strategy. If we only use data from the games where we do observe $\opp[z]$, then  our samples are biased. Instead of directly building a model of the opponent strategy, we construct a set of constraints which the strategy must satisfy in order to be consistent with the observed data.

Our first set of constraints tells us how often our opponent plays to reach any part of the game, summed across all of his hands. In particular, this tells us how often he folds given the opportunity to do so. We construct these constraints estimating how often we hold a hand $\our \in \Priv_i$ and the opponent plays to reach a public node $\pub \in \S$. Because we always observe $\our[h]$ and $\pub[h]$, we have unbiased samples of $\Pr_\sigma[\our,\pub] = \sum_{\opp} \pi^\sigma(\our,\opp,\pub)$ for each $\our \in \Priv_i$ and $\pub \in \Pub$.\footnote{We assume that our own strategy $\sigma_i$ plays to reach each history with nonzero probability.} We apply the Wilson score interval \cite{Wilson27} to calculate bounds $L_{\our,\pub}$ and $U_{\our,\pub}$ such that $L_{\our,\pub} \leq \Pr_\sigma[\our,\pub] \leq U_{\our,\pub}$ with probability at least $\gamma$, where $\gamma$ is a chosen confidence level. Noting that 
\begin{equation*}
\Pr\nolimits_\sigma[\our,\pub] = \sum_{\opp\in\Priv_{-i}} \pi_c(\our,\opp,\pub)\pi^{\sigma_i}_i(\our,\opp,\pub)\pi^{\sigma_{-i}}_{-i}(\our,\opp,\pub)
\end{equation*}
and that $\pi_c$ and $\pi^{\sigma_i}_i$ are known, these bounds thus constitute linear sequence-form constraints. We also note that our own reach probability $\pi^{\sigma_i}_i(\our,\opp,\pub)$ cannot depend on $\opp$, as we don't observe it, and thus define $\pi^{\sigma_i}_i(\our,\pub)=\pi^{\sigma_i}_i(\our,\opp,\pub)$. We can also replace this quantity with an observed estimate $\hat{\pi}^{\sigma_i}_i(\our,\pub)$, which is unbiased and better correlated with the observed probability $\hat{\Pr}_\sigma[\pub]$. This gives us sequence-form constraints
\begin{equation}
L_{\our,\pub} \leq \hat{\pi}^{\sigma_i}_i(\our,\pub) \sum_{\opp\in\Priv_{-i}} \pi_c(\our,\opp,\pub)\pi^{\feas}_{-i}(\opp,\pub) \leq U_{\our,\pub} \qquad \forall \our \in \Priv_i,\forall \pub \in \Pub
\label{eq:phinner}
\end{equation}
where $\feas \in \Sigma_{-i}$ is the constrained strategy.

Our second set of constraints tells us how often the opponent plays to reach a showdown, i.e. a terminal history $z$ such that $\pub[z] \in \Pubterm \setminus \Pubtermu$. We always observe the full history when $\pub[z] \in \Pubterm \setminus \Pubtermu$ occurs, and thus have unbiased samples of $\pi^\sigma(z)$. Again, we calculate Wilson score intervals $L_z \leq \pi^\sigma(z) \leq U_z$ with confidence $\gamma$, which become the sequence-form constraints
\begin{equation}
L_z \leq \hat{\pi}^{\sigma_i}_i(z) \pi_c(z) \pi^{\feas}_{-i}(z) \leq U_z\qquad \forall z \in Z\text{ s.t. }\pub[z] \in \Pubterm \setminus \Pubtermu
\label{eq:phterm}
\end{equation}

Together, constraints (\ref{eq:phinner}) and (\ref{eq:phterm}) are sufficient to ensure that we maximally exploit the modeled opponent in the limit. To show this, we begin by characterizing when two strategies are equivalent in terms of expected utility.

\begin{lem}
Let $\oppstrat,\oppstratp \in \Sigma_{-i}$ be two strategies such that
\begin{alignat}{2}
\sum_{\opp \in \Priv_{-i}} \pi_c(\our,\opp,\pub)\pi^{\oppstrat}_{-i}(\opp,\pub) &= \pi_c(\our,\opp,\pub)\pi^{\oppstratp}_{-i}(\opp,\pub)\qquad &&\forall \our \in \Priv_i,\forall \pub \in \Pubtermu\label{eq:eqconf}\\
\pi^{\oppstrat}_{-i}(\opp,\pub) &= \pi^{\oppstratp}_{-i}(\opp,\pub) &&\forall \opp \in \Priv_{-i}, \forall \pub \in \Pubterm \setminus \Pubtermu.
\label{eq:eqconsd}
\end{alignat}
Then for any $\ourstrat \in \Sigma_i$, it is the case that
\begin{equation}
u(\ourstrat,\oppstrat) = u(\ourstrat,\oppstratp).
\end{equation}
\label{lem:eqexp}
\end{lem}

\begin{proof}
\begin{align}
u(\ourstrat,\oppstrat) &= \sum_{z \in Z} \pi_c(z)\pi^{\ourstrat}_i(z)\pi^{\oppstrat}_{-i}(z) u(z)\\
&= \sum_{\pub \in \Pubtermu}\sum_{\our \in \Priv_i} \sum_{\opp \in \Priv_{-i}} \pi_c(\our,\opp,\pub)\pi^{\ourstrat}_i(\our,\pub)\pi^{\oppstrat}_{-i}(\opp,\pub) u(\our,\opp,\pub)\\
&\qquad + \sum_{\pub \in \Pubterm \setminus \Pubtermu}\sum_{\our \in \Priv_i} \sum_{\opp \in \Priv_{-i}} \pi_c(\our,\opp,\pub)\pi^{\ourstrat}_i(\our,\pub)\pi^{\oppstrat}_{-i}(\opp,\pub) u(\our,\opp,\pub)\\
&= \sum_{\pub \in \Pubtermu} u(\pub) \sum_{\our \in \Priv_i} \pi^{\ourstrat}_i(\our,\pub) \sum_{\opp \in \Priv_{-i}} \pi_c(\our,\opp,\pub)\pi^{\oppstrat}_{-i}(\opp,\pub)\\
&\qquad + \sum_{\pub \in \Pubterm \setminus \Pubtermu}\sum_{\our \in \Priv_i} \sum_{\opp \in \Priv_{-i}} \pi_c(\our,\opp,\pub)\pi^{\ourstrat}_i(\our,\pub)\pi^{\oppstrat}_{-i}(\opp,\pub) u(\our,\opp,\pub)\\
&= \sum_{\pub \in \Pubtermu} u(\pub) \sum_{\our \in \Priv_i} \pi^{\ourstrat}_i(\our,\pub) \sum_{\opp \in \Priv_{-i}} \pi_c(\our,\opp,\pub)\pi^{\oppstratp}_{-i}(\opp,\pub)\\
&\qquad + \sum_{\pub \in \Pubterm \setminus \Pubtermu}\sum_{\our \in \Priv_i} \sum_{\opp \in \Priv_{-i}} \pi_c(\our,\opp,\pub)\pi^{\ourstrat}_i(\our,\pub)\pi^{\oppstratp}_{-i}(\opp,\pub) u(\our,\opp,\pub)\\
&= u(\ourstrat,\oppstratp) \qedhere
\end{align}
\end{proof}

\begin{thm}
Given partial observations from $n$ games of an opponent using strategy $\sigma_{-i}$, let $\sigma_i$ be a strategy computed by running CCFR with constraints (\ref{eq:phinner}) and (\ref{eq:phterm}) for $T$ iterations. Then as $n$ and $T$ approach infinity, $\sigma_i$ converges to a best response to $\sigma_{-i}$.
\end{thm}

\begin{proof}
As $n$ goes to infinity, each $\hat{\pi}_i^{\sigma_i}(h)$ converges to $\pi_i^{\sigma_i}(h)$, and each confidence interval bound $L$ or $U$ converges to the true expectation of the estimated statistic. Thus, our constraints converge to equality constraints and any feasible strategy $\feas \in \Sigma_{-i}$ must satisfy
\begin{alignat*}{2}
\sum_{\opp \in \Priv_{-i}} \pi_c(\our,\opp,\pub)\pi^{\oppstrat}_{-i}(\opp,\pub) &= \pi_c(\our,\opp,\pub)\pi^{\feas}_{-i}(\opp,\pub)\qquad &&\forall \our \in \Priv_i,\forall \pub \in \Pubtermu\\
\pi^{\oppstrat}_{-i}(\opp,\pub) &= \pi^{\feas}_{-i}(\opp,\pub) &&\forall \opp \in \Priv_{-i}, \forall \pub \in \Pubterm \setminus \Pubtermu.
\end{alignat*}
By Lemma~\ref{eqexp} we have that $u(\sigma_i,\feas) = u(\sigma_i,\sigma_{-i})$ for any $\sigma_i \in \Sigma_i$. By Theorem~1, the average strategy $\overline{\sigma}_i^T$ converges to a best response to $\overline{\sigma}_{-i}^T$, and $\overline{\sigma}_{-i}^T$ converges to the feasible set. Thus
\begin{equation}
u(\overline{\sigma}_i^T,\sigma_{-i}) = u(\overline{\sigma}_i^T,\overline{\sigma}_{-i}^T) \geq   u(\sigma_i',\overline{\sigma}_{-i}^T) = u(\sigma_i',\sigma_{-i})
\end{equation}
and $\overline{\sigma}_i^T$ is a best response to $\sigma_{-i}$.
\end{proof}

\end{document}